\documentclass[a4paper,leqno]{amsart}

\usepackage{latexsym}
\usepackage[english]{babel}
\usepackage{fancyhdr}
\usepackage[mathscr]{eucal}
\usepackage{amsmath}
\usepackage{mathrsfs}
\usepackage{mathabx}
\usepackage{amsthm}
\usepackage{amsfonts}
\usepackage{amssymb}
\usepackage{amscd}
\usepackage{bbm}
\usepackage{graphicx}
\usepackage{graphics}
\usepackage{latexsym}
\usepackage{color}

\usepackage{scalerel,stackengine}
\stackMath
\newcommand\reallywidehat[1]{%
\savestack{\tmpbox}{\stretchto{%
  \scaleto{%
    \scalerel*[\widthof{\ensuremath{#1}}]{\kern-.6pt\bigwedge\kern-.6pt}%
    {\rule[-\textheight/2]{1ex}{\textheight}}
  }{\textheight}%
}{0.5ex}}%
\stackon[1pt]{#1}{\tmpbox}%
}

\newcommand{\ud}{\mathrm{d}}

\newcommand{\ii}{\mathrm{i}}
\newcommand{\cH}{\mathcal{H}}

\theoremstyle{plain}
\newtheorem{theorem}{Theorem}[section]
\newtheorem{lemma}[theorem]{Lemma}
\newtheorem{corollary}[theorem]{Corollary}
\newtheorem{proposition}[theorem]{Proposition}

\theoremstyle{definition}

\newtheorem{remark}[theorem]{Remark}
\newtheorem*{remark*}{Remark}

\numberwithin{equation}{section}

\begin{document}

\title[Spectrum of the 2+1 fermionic trimer with contact interactions]
{Spectral analysis of the 2+1 fermionic trimer with contact interactions}

\author[S.~Becker]{Simon Becker}
\address[S.~Becker]{Department of Applied Mathematics and Theoretical Physics, University of Cambridge \\ Wilbeforce Rd, Cambridge \\ CB3 0WA (UK).}
\email{simon.becker@damtp.cam.ac.uk}
\author[A.~Michelangeli]{Alessandro Michelangeli}
\address[A.~Michelangeli]{International School for Advanced Studies -- SISSA \\ via Bonomea 265 \\ 34136 Trieste (Italy).}
\email{alessandro.michelangeli@sissa.it}
\author[A.~Ottolini]{Andrea Ottolini}
\address[A.~Ottolini]{Department of Mathematics, Stanford University \\ 450 Serra Mall, Stanford CA 94305 (USA).}
\email{ottolini@stanford.edu}

\begin{abstract}
We qualify the main features of the spectrum of the Hamiltonian of point interaction for a three-dimensional quantum system consisting of three point-like particles, two identical fermions, plus a third particle of different species, with two-body interaction of zero range. For arbitrary magnitude of the interaction, and arbitrary value of the mass parameter (the ratio between the mass of the third particle and that of each fermion) above the stability threshold, we identify the essential spectrum, localise the discrete spectrum and prove its finiteness, qualify the angular symmetry of the eigenfunctions, and prove the increasing monotonicity of the eigenvalues with respect  to the mass parameter. We also demonstrate the existence or absence of bound states in the physically relevant regimes of masses.
\end{abstract}

\date{\today}

\subjclass[2000]{}

\thanks{Partially supported by the 2014-2017 MIUR-FIR grant ``\emph{Cond-Math: Condensed Matter and Mathematical Physics}'' code RBFR13WAET (S.B., A.M., A.O.), by the EPSRC grant EP/L016516/1 for the University of Cambridge CDT, the CCA (S.B.), and by a 2017 visiting research fellowship  at the International Center for Mathematical Research CIRM, Trento (A.M.).}
\maketitle


\section{Introduction}

In this work we consider the three-dimensional quantum system consisting of three point-like particles, two identical fermions, plus a third particle of different species and mass ratio $m$ with respect to the mass of the fermions, where the inter-particle interaction has exactly zero range. Customarily one refers to this system as the 2+1 fermionic trimer with `contact interaction'; in the special case where the interaction has infinite scattering length one speaks of a 2+1 fermionic trimer `at unitarity'. In practice, because of the exclusion principle, the interaction is only present between each fermion and the third particle.

For the Hamiltonian that models this system we consider the so-called maximal Ter-Martirosyan--Skornyakov realisation $H_\alpha$, for given interaction scattering length $\alpha^{-1}$ (in suitable units) and mass parameter $m>m^*$, the critical threshold for stability ($m^*\sim(13.607)^{-1}$). The Hamiltonian $H_\alpha$, whose precise construction will be recalled in Section \ref{sec:setting_and_main_results}, is a self-adjoint operator on the internal Hilbert space
\begin{equation}
 \cH\;:=\;L^2_{\mathrm{f}}(\mathbb{R}^3\times\mathbb{R}^3,\ud x_1\ud x_2)\,,
\end{equation}
where $x_j$ is the three-dimensional relative variable of the $j$-th fermion and the third particle and the subscript `f' stands for the fermionic subspace of the corresponding $L^2$-space, namely the square-integrable functions $\psi(x_1,x_2)$ that are anti-symmetric under exchange $x_1\leftrightarrow x_2$ -- of course the total Hilbert space is given by tensoring $\cH$ itself with another copy of $L^2(\mathbb{R}^3)$ corresponding to the centre-of-mass degree of freedom. The elements $\psi$ of the domain of $H_\alpha$ are qualified by a boundary condition that informally reads
\begin{equation}\label{eq:BP_condition}
 \psi(x_1,x_2)\;\sim\;\xi(x_1)\Big(\frac{1}{|x_2|}+\alpha\Big)\qquad \textrm{as }\;x_2\to 0
\end{equation}
for some $\xi$ that we are going to specify later, which is a typical Bethe-Peierls contact condition \cite{Bethe_Peierls-1935,Bethe_Peierls-1935-np}, that is, the asymptotics for the low-energy quantum scattering through an interaction only supported at the origin and $s$-wave scattering length equal to $-(4\pi\alpha)^{-1}$.

The study of this and similar few-body systems with contact interaction has a long history and a wide literature throughout the last eight decades -- a concise retrospective may be found in \cite[Section 2]{MO-2016}. Recent progress in cold atom physics made the subject topical, thanks to the possibility of tuning the effective scattering length by means of a magnetically induced Feshbach resonance \cite[Section 5.4.2]{pethick02}, thus making the zero-range idealisation particularly realistic, above all at unitarity  \cite{Braaten-Hammer-2006,Castin-Werner-2011_-_review}. The 2+1 fermionic system is an actual building block for hetheronuclear mixtures with inter-species contact interaction -- see \cite[Section 1]{MP-2015-2p2} for an outlook.

Mathematically, the correct implementation of the `physical' condition \eqref{eq:BP_condition}, so as to identify a well-defined self-adjoint Hamiltonian, is a problem originally set up by Minlos and Faddeev \cite{Minlos-Faddeev-1961-1,Minlos-Faddeev-1961-2}, which required a long time to be fully understood and solved -- and in fact is still open for other fundamental configurations of the type $N+M$ ($N$ identical fermions of one type and $M$ of another type with inter-species contact interaction).

For the 2+1 fermionic model of interest, the rigorous construction of the Hamiltonian $H_\alpha$ for $m>m^*$, together with the precise determination of $m^*$ and the proof of the self-adjointness and the semi-boundedness from below of $H_\alpha$, was done in the work \cite{CDFMT-2012} by Correggi, Dell'Antonio, Finco, Michelangeli, and Teta, by means of quadratic form techniques for contact interactions \cite{Teta-1989,Finco-Teta-2012}. 
A previous thorough investigation by Minlos \cite{Minlos-1987,Minlos-2012-preprint_30sett2011,Minlos-2012-preprint_1nov2012,Minlos-RusMathSurv-2014} and Minlos and Shermatov \cite{Minlos-Shermatov-1989}, based instead on the Kre{\u\i}n-Vi\v{s}ik-Birman self-adjoint extension theory \cite{GMO-KVB2017}, had the virtue of showing that associated with the boundary condition \eqref{eq:BP_condition}, and depending whether $m$ is large or small, there is only one or an infinity of distinct self-adjoint realisations -- the Hamiltonians  `Ter-Martirosyan--Skornyakov' type for the 2+1 system -- however with a flaw in the treatment of the so-called `space of charges', as was later found by Michelangeli and Ottolini in \cite{MO-2016}. The identification of $H_\alpha$ as the highest (Friedrichs-type) among all the Ter-Martirosyan--Skornyakov self-adjoint Hamiltonians was finally made by Michelangeli and Ottolini in \cite{MO-2017}.

In this work we continue the investigation of the Hamiltonian $H_\alpha$ by qualifying a number of key features of its spectrum. In the physics literature, where the contact condition \eqref{eq:BP_condition} is typically implemented at a formal level, that is, ignoring the precise domain issues and the well-posedness of the operator, it nevertheless gives rise to manageable numerical schemes for the computation of ground state energy and higher eigenvalues \cite{Kartavtsev-Malykh-2007,Endo-Naidon-Ueda-2011,Kartavtsev-Malykh-2016,Kartavtsev-Malykh-2016-proc}. There is no analogue at a rigorous level, except for the work \cite{michelangeli-schmidbauer-2013} by Michelangeli and Schmidbauer, where the numerical computation of the ground state energy of $H_\alpha$ was set up on the eigenvalue problem emerging from the rigorous definition of  $H_\alpha$.

Informally speaking, the main findings of the present work are: we determine exactly the essential spectrum of $H_\alpha$ and its bottom, which is zero for repulsive interaction ($\alpha\geqslant 0$) and strictly negative for attractive interactions ($\alpha<0$); we localise a non-trivial energy window that contains entirely the discrete spectrum of $H_\alpha$; we prove that the discrete spectrum, when it is non-empty, consists of finitely many eigenvalues; we prove that $H_\alpha$ does admit eigenvalues below the essential spectrum in a physically relevant regime of masses $m$; and we prove that such negative eigenvalues are monotone increasing with $m$.

The precise setting for the 2+1 model, the rigorous formulation of our results, and their general discussion is presented in Section \ref{sec:setting_and_main_results}. The subsequent Sections contain the proofs and an amount of material needed from the previous literature.

\textbf{Notation.} Besides an amount of standard notation, we shall denote by $\mathcal{D}(H)$, resp., $\mathcal{D}[H]$, the operator domain, resp., the form domain of a self-adjoint operator $H$ on Hilbert space;
we shall denote with $\rightharpoonup$ the Hilbert space weak convergence; we shall omit for compactness the `almost everywhere' specification for pointwise identities between functions in Lebesgue spaces; we shall denote by $\mathbbm{1}$, resp., by $\mathbbm{O}$, the identity and the null operator on any of the considered Hilbert spaces; we shall indicate the Fourier transform by $\widehat{\phi}$ or $\mathcal{F}\phi$ with the convention $\widehat{\phi}(p)=(2\pi)^{-\frac{d}{2}}\int_{\mathbb{R}^d}e^{-\ii p x}\phi(x)\ud x $; we shall denote sequences by $(\phi_n)_n$; we shall use $\dotplus$, resp., $\oplus$, for the direct sum, resp., the direct orthogonal sum; we shall write $A\lesssim B$ for $A\leqslant\mathrm{const.}\,B$ when the constant does not depend on the other relevant parameters or variables of both sides of the inequality; for $x\in\mathbb{R}^d$ we shall write $\langle x\rangle:=(1+x^2)^{\frac{1}{2}}$.
%

\section{Setting and main results}\label{sec:setting_and_main_results}

Acting on the Hilbert space $\cH=L^2_{\mathrm{f}}(\mathbb{R}^3\times\mathbb{R}^3,\ud x_1\ud x_2)$, let us introduce the self-adjoint operator $H_\alpha$, for given $\alpha\in\mathbb{R}\cup\{\infty\}$ -- here we follow closely \cite{CDFMT-2012,MO-2016,MO-2017}.

Let us observe, preliminarily, that after removing the centre of mass, the free Hamiltonian of a three-dimensional system of two identical fermions of unit mass in relative positions $x_1,x_2$ with respect to a third particle of different species and with mass $m$ is the operator
\begin{equation}\label{eq:Hfree}
 H_{\mathrm{free}}\;:=\;-\Delta_{x_1}-\Delta_{x_2}-{\textstyle\frac{2}{m+1}}\,\nabla_{x_1}\cdot\nabla_{x_2}
\end{equation}
acting on $\cH$ with domain of self-adjointness consisting of the $H^2$-functions in $\cH$. (The three-body free Hamiltonian $H_{2+1}$ in the coordinates $(x_1,x_2,x_{\mathrm{cm}})$, where $x_{\mathrm{cm}}$ is the centre-of-mass variable, acts as $H_{2+1}=\frac{1+m}{2m}H_{\mathrm{free}}-\frac{1}{\,2(m+2)}\Delta_{x_{\mathrm{cm}}}$.)

We set for convenience
\begin{equation}\label{eq:defmunu}
 \mu\;:=\;\frac{2}{m+1}\,,\qquad\nu\;:=\;\frac{m(m+2)}{\:(m+1)^2}\;=\;1-\frac{\mu^2}{4}\,,
\end{equation}
and for arbitrary $\lambda>0$ we consider the linear map $\xi\mapsto T_\lambda\xi$ defined by
\begin{equation}\label{eq:Tlambda}
\widehat{(T_\lambda\,\xi)}(p)\;:=\;2\pi^2\sqrt{\nu p^2+\lambda}\;\widehat{\xi}(p)+\int_{\mathbb{R}^3}\frac{\widehat{\xi}(q)}{p^2+q^2+\mu p\cdot q+\lambda}\,\ud q\,,\quad p\in\mathbb{R}^3\,.
\end{equation}
One can see (see Proposition \ref{prop:regpropTlambda} in the following) that $T_\lambda$  maps continuously $H^s(\mathbb{R}^3)$ into $H^{s-1}(\mathbb{R}^3)$ for any $s\in(-\frac{1}{2},\frac{3}{2})$: we shall refer to it as the `\emph{charge operator}'. Moreover, for any $\xi\in H^{-\frac{1}{2}}(\mathbb{R}^3)$ and $\lambda>0$ we define
\begin{equation}\label{eq:u_xi}
 \widehat{u_\xi^\lambda}(p_1,p_2)\;:=\;\frac{\widehat{\xi}(p_1)-\widehat{\xi}(p_2)}{p_1^2+p_2^2+\mu\,p_1\cdot p_2+\lambda}\,,\qquad p_1,p_2\in\mathbb{R}^3\,.
\end{equation}
One can also see \cite[Prop.~1]{MO-2017} that $u_\xi^\lambda\in\cH$ and $u_\xi^\lambda=u_\eta^\lambda$ $\Rightarrow$ $\xi=\eta$.

In this context we shall refer to $H^{-\frac{1}{2}}(\mathbb{R}^3)$ and to the $\xi$'s in $H^{-\frac{1}{2}}(\mathbb{R}^3)$ as, respectively, the `\emph{space of charges}' and the `\emph{charges}' for the system under consideration. This jargon dates back from \cite{DFT-1994} and is chosen in analogy with electrostatics: indeed, \eqref{eq:u_xi} implies that, distributionally,
\[
 \begin{split}
  \big((-\Delta_{x_1}-\Delta_{x_2}-{\textstyle\frac{2}{m+1}}&\,\nabla_{x_1}\cdot\nabla_{x_2}+\lambda)u_\xi^\lambda\big)(x_1,x_2)\;= \\
  &=\;(2\pi)^{\frac{3}{2}}\big(\xi(x_1)\delta(x_2)-\xi(x_2)\delta(x_1)\big)\,,
 \end{split}
\]
thus $u_\xi^\lambda$ can be considered as the analogue of an electrostatic potential generated by the charge distribution $\xi$.

Next, we introduce the threshold mass $m^*\approx(13.607)^{-1}$ as the unique root of $\Lambda(m)=1$, where 
\begin{equation}\label{eq:Lambdam}
\Lambda(m)\;:=\;{\textstyle\frac{2}{\pi}}(m+1)^2\Big(\frac{1}{\sqrt{m(m+2)}}-\arcsin\frac{1}{m+1}\Big)\,.
\end{equation}
Indeed, $(0,+\infty)\ni m\mapsto\Lambda(m)$ is a positive, smooth, monotone decreasing function with range $(0,+\infty)$. It is often referred to as the `\emph{Efimov transcendental function}'.

The operator $H_\alpha$ is then defined for $m>m^*$ as follows: its domain is
\begin{equation}\label{eq:DHab}
\begin{split}
\mathcal{D}(H_{\alpha})\;&:=\;\left\{\,g=F^\lambda+u^\lambda_\xi\left|\!
\begin{array}{c}
F^\lambda \in H^2_{\mathrm{f}}(\mathbb{R}^3\times\mathbb{R}^3)\,,\\
\xi\in H^{\frac{1}{2}}(\mathbb{R}^3)\,,\;(T_\lambda+\alpha\mathbbm{1})\xi\in H^{\frac{1}{2}}(\mathbb{R}^3)\,, \\
\textrm{plus the boundary conditions $(\textsc{tms}')$}
\end{array}\!\!\!\right.\right\} \\
(\textsc{tms}') & \quad \int_{\mathbb{R}^3}\widehat{F^\lambda}(p_1,p_2)\,\ud p_2\;=\;\big((T_\lambda+\alpha)\,\xi\big)^{\!\!\widehat{\;\;\;}}(p_1)\,,
\end{split}
\end{equation}
and its action is
\begin{equation}\label{eq:action_DHab}
\begin{split}
(H_{\alpha}+\lambda\mathbbm{1})\,g\;:=&\;(H_{\mathrm{free}}+\lambda\mathbbm{1})\,F^\lambda \\
=&\;\big(-\Delta_{x_1}-\Delta_{x_2}-{\textstyle\frac{2}{m+1}}\nabla_{x_1}\cdot\nabla_{x_2}+\lambda\big)\,F^\lambda\,.
\end{split}
\end{equation}

It is worth remarking that the decomposition \eqref{eq:DHab} of the domain of $H_\alpha$ depends on the chosen $\lambda>0$, but the domain itself, and the operator action, does not (Section \ref{sec:lambda_arbitrary}).
In fact, $\lambda$ represents a fictitious artefact that pops up ubiquitously for the mere reason that it is much more manageable to construct the shifted operator $H_{\alpha}+\lambda\mathbbm{1}$ and to study its properties \cite{CDFMT-2012,MO-2016,MO-2017}, from which one infers at once the analogous picture for $H_\alpha$ itself. On the other hand, it is important to remember that for each chosen $m>m^*$ formulas \eqref{eq:DHab}-\eqref{eq:action_DHab} define a different operator $H_\alpha$, even if in our notation we do not carry over the explicit dependence on $m$.

Let us comment on the relevance of the operator $H_\alpha$ for the 2+1 fermionic system with contact interaction. Let us consider the restriction of $H_{\mathrm{free}}$ to functions that vanish in the vicinity of the `\emph{coincidence hyperplanes}' $\Gamma_1$ and $\Gamma_2$, where
\begin{equation}
\Gamma_j\;:=\;\{(x_1,x_2)\in\mathbb{R}^3\times\mathbb{R}^3\,|\,x_j=0\}\,,\qquad j\in\{1,2\}\,.
\end{equation}
More precisely, let us consider the operator
\begin{equation}\label{eq:Hring_2+1}
\begin{split}
\mathring{H}\;&:=\;-\Delta_{x_1}-\Delta_{x_2}-{\textstyle\frac{2}{m+1}}\,\nabla_{x_1}\cdot\nabla_{x_2} \\
\mathcal{D}(\mathring{H})\;&:=\;H^2_0((\mathbb{R}^3\times\mathbb{R}^3)\!\setminus\!(\Gamma_1\cup\Gamma_2))\cap\cH
\end{split}
\end{equation}
on $\cH$. (We recall that 
\begin{equation*}
H^2_0((\mathbb{R}^3\times\mathbb{R}^3)\!\setminus\!(\Gamma_1\cup\Gamma_2))\;=\;\overline{C^\infty_0((\mathbb{R}^3\times\mathbb{R}^3)\!\setminus\!(\Gamma_1\cup\Gamma_2))}^{\|\,\|_{H^2}}\,,
\end{equation*}
which is the space used in \eqref{eq:Hring_2+1}.) It can be easily seen that $\mathring{H}$ is densely defined, closed, positive, and symmetric on $\cH$. It fails to be self-adjoint (in fact, it has infinite deficiency indices), but it admits self-adjoint extensions, for it is semi-bounded from below. These are restrictions of the adjoint $\mathring{H}^*$, and in fact the functions of the form \eqref{eq:u_xi} can be proved to span $\ker(\mathring{H}^*+\lambda\mathbbm{1})$, thus
\begin{equation}\label{eq:kerHstar}
 \mathring{H}^*u_\xi^\lambda\;=\;-\lambda u_\xi^\lambda\,.
\end{equation}
Each such extension is a legitimate Hamiltonian for the 2+1 interacting system, with an interaction that then is only supported at the coincidence hyperplanes.

The latter circumstance is precisely what happens for $H_\alpha$. Let us indeed recall the following known properties.

\begin{theorem}\label{thm:properties_of_Halpha} Let $m>m^*$.
\begin{itemize}
 \item[(i)] $H_\alpha$ is self-adjoint on $\cH$.
 \item[(ii)] $H_\alpha$ is semi-bounded from below, with
 \begin{equation}\label{eq:inf_spec_Halpha}
  \begin{array}{lll}
  H_\alpha\;\geqslant\;\mathbbm{O}  & & \textrm{if }\;\alpha\geqslant 0 \\
  H_\alpha\;\geqslant\;-\frac{\alpha^2}{4\pi^4(1-\Lambda(m)^2)}\,\mathbbm{1} & &  \textrm{if }\;\alpha < 0\,.
  \end{array}
 \end{equation}
 \item[(iii)] $H_\alpha$ is an extension of $\mathring{H}$.
 \end{itemize}
\end{theorem}

Theorem \ref{thm:properties_of_Halpha} was established first in \cite{CDFMT-2012} by means of the quadratic form theory, and later in \cite{MO-2017}, by means of the Kre{\u\i}n-Vi\v{s}ik-Birman self-adjoint extension theory, in either case under the operational restriction that when $\alpha<0$ the parameter $\lambda$ in formulas \eqref{eq:DHab}-\eqref{eq:action_DHab} must be taken sufficiently large, more precisely $\lambda>\lambda_0$ with
\begin{equation*}
 \lambda_0\;:=\;\frac{2|\alpha|}{4\pi^2(1-\Lambda(m))}
\end{equation*}
(see \cite[Prop.~3.1]{CDFMT-2012} and \cite[Prop.~6]{MO-2017}). Since for the present discussion the full arbitrariness of $\lambda$ is crucial, and to our knowledge  it was never shown in the previous literature that the restriction to $\lambda>\lambda_0$ can be by-passed, we give evidence of this fact in Section \ref{sec:lambda_arbitrary}. Thus from now on we can indeed consider the definition of $H_\alpha$ and the properties of Theorem \ref{thm:properties_of_Halpha} valid irrespectively of the chosen $\lambda>0$.

We observe that \eqref{eq:inf_spec_Halpha} only provides a (non optimal) lower bound to $H_\alpha$ when $\alpha<0$. In fact, the lower bound determined in \cite{CDFMT-2012} was slightly less efficient and had the form
\begin{equation}\label{eq:old_lower_bound}
 H_\alpha\;\geqslant\;-\textstyle\frac{\alpha^2}{4\pi^4(1-\Lambda(m))^2}\,\mathbbm{1} \qquad\qquad (\alpha < 0)\,.
\end{equation}
A more careful analysis shows that in fact \eqref{eq:old_lower_bound} can be ameliorated to \eqref{eq:inf_spec_Halpha}. We discuss this in Section \ref{sec:improved_lower_bound}.

$H_\alpha$ displays one more fundamental property that makes it a most natural Hamiltonian of contact interaction for the 2+1 fermionic system. It is encoded in the boundary condition $(\textsc{tms}')$ of \eqref{eq:DHab}. This is an identity in $H^{\frac{1}{2}}(\mathbb{R}^3)$ linking the `regular' $H^2$-part $F^\lambda$ of $g\in\mathcal{D}(H_\alpha)$ with the `singular' part $u_\xi^\lambda$ of $g$. It  can be better understood if one recalls \cite[Lemma 4]{MO-2016} that
\begin{equation}\label{eq:TMS_pre-asymptotics}
\int_{\substack{ \\ \,p_2\in\mathbb{R}^3 \\ \! |p_2|<R}}\;\widehat{u_\xi^\lambda}(p_1,p_2)\,\ud p_2\;=\;4\pi\widehat{\xi}(p_1) R-\widehat{(T_\lambda\,\xi)}(p_1)+o(1)\quad\textrm{as $R\to +\infty$}\,;
\end{equation}
then, combining \eqref{eq:TMS_pre-asymptotics} with \eqref{eq:DHab} yields the following.

\begin{proposition}\label{prop:TMS_condition}
  Let $m>m^*$. Let $g\in\mathcal{D}(H_\alpha)$ and let $\xi$ be the corresponding charge of $g$, according to the decomposition \eqref{eq:DHab}. Then $g$ satisfies the `\emph{Ter-Martirosyan--Skornyakov condition}'
  \begin{equation}\tag{\textsc{tms}}\label{eq:TMS_pointwise}
   \int_{\substack{ \\ \,p_2\in\mathbb{R}^3 \\ \! |p_2|<R}}\;\widehat{g}(p_1,p_2)\,\ud p_2\;=\;(4\pi R+\alpha)\,\widehat{\xi}(p_1)+o(1)\quad\textrm{ as }R\to +\infty\,.
  \end{equation}
\end{proposition}

The boundary condition \eqref{eq:TMS_pointwise} (together with its mirror version when $p_1\leftrightarrow p_2$, using the anti-symmetry of $g$) is the Fourier transform counterpart of the Bethe-Peierls contact condition \eqref{eq:BP_condition}. It is an ultra-violet asymptotics when each fermion comes on top of the third particle. Analogously to \eqref{eq:BP_condition}, as recognised for the first time by Ter-Martirosyan and Skornyakov in \cite{TMS-1956}, it encodes the presence of an interaction active only for the coincidence configurations and with $s$-wave scattering length $-(4\pi\alpha)^{-1}$.

The fact that $H_\alpha$ is a self-adjoint extension of $\mathring{H}$ and models an interaction supported on the hyperplanes $\Gamma_1$ and $\Gamma_2$ which satisfies the physically meaningful asymptotics \eqref{eq:TMS_pointwise} makes $H_\alpha$ the \emph{natural Hamiltonian of contact interaction for the 2+1 fermionic trimer}.

For completeness, let us include the quadratic form of the operator $H_\alpha$, which is given by
\begin{equation}\label{eq:DHab_form}
\begin{split}
\mathcal{D}[H_{\alpha}]\;&=\;\left\{g=F^\lambda+u^\lambda_\xi\left|\!\!
\begin{array}{c}
F^\lambda \in H^1_{\mathrm{f}}(\mathbb{R}^3\times\mathbb{R}^3),\;\xi\in H^{\frac{1}{2}}(\mathbb{R}^3)
\end{array}\!\!\!\!\right.\right\} \\
H_{\alpha}[F^\lambda+u^\lambda_\xi]\;&=\;\lambda\|F^\lambda\|_\cH^2-\lambda\|F^\lambda+u_\xi^\lambda\|_\cH^2+H_\mathrm{free}[F^\lambda] \\
& \qquad\qquad+2\!\int_{\mathbb{R}^3}\overline{\,\widehat{\xi}(p)}\,\big(\widehat{(T_\lambda\xi)}(p)+\alpha\,\widehat{\xi}(p)\big)\,\ud p\,.
\end{split}
\end{equation}

From all the above considerations it is also clear that $H_{\alpha=\infty}$ is the self-adjoint free Hamiltonian with domain the $H^2$-functions of $\cH$: in this case, the interaction is absent.

Before proceeding with the goal of our investigation and our main results, let us quickly mention the following general problem, for a thorough discussion of which we refer to \cite{CDFMT-2015,MO-2016,MO-2017} and the references therein. Out of the vast variety of self-adjoint extensions of $\mathring{H}$, each of which is interpreted as a Hamiltonian of contact interaction, the physical boundary condition \eqref{eq:TMS_pointwise} selects uniquely $H_\alpha$ for sufficiently large $m$, whereas it selects an \emph{infinite} sub-family of extensions, all of Ter-Martirosyan--Skornyakov type, when $m$ is sufficiently small (and above $m^*$, in the present discussion). In fact, the very meaning of  \eqref{eq:TMS_pointwise} as a condition of self-adjointness need be made precise, for  \eqref{eq:TMS_pointwise} is only a \emph{point-wise} asymptotics, as opposite to the \emph{functional identity} $(\textsc{tms}')$ of \eqref{eq:DHab}: this point of view is developed in  \cite{MO-2016,MO-2017}, where the appropriate set of functional conditions of self-adjointness of Ter-Martirosyan--Skornyakov type is discussed. When an infinite multiplicity of TMS extensions arises, the present operator $H_\alpha$ turns out to be the \emph{highest}, in the sense of ordering of self-adjoint operators, all  other extensions being qualified by additional asymptotics, hence additional physics, at the triple coincidence point of the trimer. The study of the other TMS extensions is at a relatively early stage and in this work we only consider $H_\alpha$. The non-TMS extensions are in a sense non-physical, as they correspond to non-local boundary conditions: for them there is a natural classification in therms of the Kre{\u\i}n-Vi\v{s}ik-Birman extension theory (see, e.g., \cite[Eq.~(39)-(40)]{MO-2017}.

Let us present now the main results of this work. As mentioned in the Introduction, we qualify an amount of spectral properties for the Hamiltonian $H_\alpha$.

First, we identify the essential spectrum and we localise a region containing the whole discrete spectrum. The picture changes from the case of repulsive interaction ($\alpha>0$) or unitarity ($\alpha=0$) to the case of attractive interaction ($\alpha<0$).

\begin{theorem}[Essential spectrum]\label{thm:ess-spec}
Let $m>m^*$. The essential spectrum of $H_\alpha$ is the set
\begin{equation}
 \begin{array}{lll}
  \sigma_{\mathrm{ess}}(H_\alpha)\;=\;[0,+\infty) & & \textrm{if }\;\alpha\geqslant 0 \\
  \sigma_{\mathrm{ess}}(H_\alpha)\;=\;[-\frac{\alpha^2}{4\pi^4},+\infty) & & \textrm{if }\;\alpha<0\,.
 \end{array}
\end{equation}
As a consequence of this and of Theorem \ref{thm:properties_of_Halpha}(ii),
\begin{equation}\label{eq:where_disc_spec}
 \begin{array}{lll}
  \sigma_{\mathrm{disc}}(H_\alpha)\;=\;\emptyset & & \textrm{if }\;\alpha\geqslant 0 \\
  \sigma_{\mathrm{disc}}(H_\alpha)\;\subset\;[-\frac{\alpha^2}{4\pi^4(1-\Lambda(m)^2)}, -\frac{\alpha^2}{4\pi^4})& & \textrm{if }\;\alpha<0\,.
 \end{array}
\end{equation}
\end{theorem}

Next, we show that even when the discrete spectrum is not a priori empty, it can only consist of finitely many eigenvalues. This excludes the accumulation of eigenvalues to the bottom of the essential spectrum, a spectral feature that for three-body interacting systems is also known as the Efimov effect \cite{Braaten-Hammer-2006,Petrov-Ultracoldgases-LH2010}.

\begin{theorem}[Finiteness of the discrete spectrum]\label{thm:disc-spec-finite}
 Let $m>m^*$ and $\alpha<0$. Then $\sigma_{\mathrm{disc}}(H_\alpha)$ consists of only finitely many eigenvalues (each of which has of course finite multiplicity). 
\end{theorem}

That the portion of discrete spectrum below $-\frac{\alpha^2}{4\pi^4}$ is finite and that $\inf\sigma_{\mathrm{ess}}(H_\alpha)=-\frac{\alpha^2}{4\pi^4}$ when $\alpha<0$ was recently proved by Yoshitomi \cite{Yoshitomi_MathSlov2017} only for very large masses. Our Theorems \ref{thm:ess-spec} and \ref{thm:disc-spec-finite} are valid for arbitrary mass parameter above the threshold $m^*$.

We then move to the question of the existence of bound states for $H_\alpha$ for an attractive contact interaction ($\alpha<0$).

In a previous work by one of us in collaboration with Schmidbauer \cite{michelangeli-schmidbauer-2013} the eigenvalue problem for $H_\alpha$ was turned into an equivalent eigenvalue problem for a radial bounded operator whose much more manageable numerical diagonalisability gave evidence that $H_\alpha$ admits a ground state for any value $m\in(m^*,m^{**})$ of the mass parameter. The threshold $m^*\approx(13.607)^{-1}$ is the one introduced in \eqref{eq:Lambdam}. The threshold $m^{**}\approx(8.62)^{-1}$, that is implicitly given by the exact formula \eqref{eq:m**root}, appears to have a central relevance in the family of Ter-Martirosyan-Skornyakov Hamiltonians of contact interaction for the 2+1 fermionic trimer: it is indeed the threshold that is expected to discriminate between the presence of a unique TMS Hamiltonian ($m>m^{**}$) and an infinite multiplicity of TMS Hamiltonians ($m<m^{**}$) \cite{CDFMT-2015,MO-2017}.

Our next result proves that in a slightly larger mass regime, $m\in(m^*,m_\star)$, there exist indeed bound states below the bottom of the essential spectrum. In addition, we prove that the corresponding eigenfunction is of the form $u_\xi^\lambda$ where the charge $\xi$ carries a specific angular symmetry, and we also prove that there are no bound states for sufficiently large masses.

\begin{theorem}[Existence and non-existence of eigenvalues]\label{thm:EVexist} Let $\alpha<0$.
\begin{itemize}
 \item[(i)] All eigenvalues $-E$ of $H_\alpha$ below $\inf\sigma_{\mathrm{ess}}(H_\alpha)$ have eigenfunctions of the form $u_\xi^{E}$ defined in \eqref{eq:u_xi}, where the charge $\xi$ belongs to the eigenspace of angular momentum quantum number $\ell=1$.
 \item[(ii)] $H_\alpha$ has eigenvalues below $\inf\sigma_{\mathrm{ess}}(H_\alpha)$ for any value $m\in(m^*,m_\star)$ of the mass parameter, the threshold $m_\star\approx(8.587)^{-1}$ being defined in \eqref{eq:defmsubstar} and satisfying $m_\star>m^{**}$. 
 \item[(iii)] There exists a mass $M_\star\leqslant(2.617)^{-1}$ such that when $m>M_\star$ the operator $H_\alpha$ has no eigenvalues below $\inf\sigma_{\mathrm{ess}}(H_\alpha)$.  
\end{itemize}
\end{theorem}

Our last result concerns the behaviour of the eigenvalues in the discrete spectrum of $H_\alpha$ in terms of the mass parameter $m$. The analytical/numerical discussion of the above-mentioned work \cite{michelangeli-schmidbauer-2013} showed that the ground state of $H_\alpha$ is monotone increasing in $m$. Here we establish the following.

\begin{theorem}[Monotonicity of eigenvalues]\label{thm:EV_monotonicity}
  Let $m>m^*$ and $\alpha<0$. The eigenvalues of $H_\alpha$ below $\inf\sigma_{\mathrm{ess}}(H_\alpha)$ are strictly monotone increasing in the mass parameter $m$.
\end{theorem}

Let us conclude by remarking that our findings, based on the rigorous model for $H_\alpha$ presented in this Section, are consistent with the numerical and experimental evidence of the physical literature on the binding properties of the 2+1 fermionic trimer -- see, e.g., \cite[Sec.~2-3]{Kartavtsev-Malykh-2007}, \cite[Sec.~3]{Endo-Naidon-Ueda-2011}, as well as the above-mentioned works
\cite{Braaten-Hammer-2006,Petrov-Ultracoldgases-LH2010,Castin-Werner-2011_-_review,Kartavtsev-Malykh-2016,Kartavtsev-Malykh-2016-proc}, based on very similar, yet only formal models.

\section{Arbitrariness of the parameter $\lambda>0$}\label{sec:lambda_arbitrary}

As announced in Section \ref{sec:setting_and_main_results}, let us show here that the definition \eqref{eq:DHab}-\eqref{eq:action_DHab} of the operator $H_\alpha$, as well as the definition \eqref{eq:DHab_form} of the quadratic form of $H_\alpha$, are independent of $\lambda$, and hence that Theorem \ref{thm:properties_of_Halpha} and Proposition \ref{prop:TMS_condition} that we recalled from the previous literature, as well as our new results, are all valid irrespective of the auxiliary parameter $\lambda$ chosen in the definition.

Let us start with the operator's definition.

\begin{lemma}\label{lem:samedomainshift-1}
 Assume that the operator $H_\alpha$ is only defined by \eqref{eq:DHab}-\eqref{eq:action_DHab} for a given $\lambda>0$. Then $H_\alpha$ satisfies \eqref{eq:DHab}-\eqref{eq:action_DHab} for any other $\lambda'>0$.
\end{lemma}

\begin{proof}
Let us decompose a generic $g\in\mathcal{D}(H_\alpha)$ as $g=F^{\lambda}+u^{\lambda}_\xi$ according to \eqref{eq:DHab}.

Now for another $\lambda'>0$ let
\[
 F^{\lambda'}\;:=\;F^{\lambda}+u^{\lambda}_\xi-u^{\lambda'}_\xi\,,
\]
where functions of the form $u_\xi$ are defined in \eqref{eq:u_xi}. Since $u^{\lambda}_\xi-u^{\lambda'}_\xi\in H^2_\mathrm{f}(\mathbb{R}^3\times\mathbb{R}^3)$, $F^{\lambda'}$ too belongs to $H^2_\mathrm{f}(\mathbb{R}^3\times\mathbb{R}^3)$. Thus, the identity
\[
 g\;=\;F^{\lambda}+u^{\lambda}_\xi\;=\;F^{\lambda'}+u^{\lambda'}_\xi
\]
shows that the decomposition of $g$ into a regular part in $H^2_\mathrm{f}(\mathbb{R}^3\times\mathbb{R}^3)$ and a singular part of the type $u_\xi$ is valid also for $\lambda'$.

It remains to show that for the new parameter $\lambda'$ both the condition $(T_{\lambda'}+\alpha\mathbbm{1})\xi\in H^{\frac{1}{2}}(\mathbb{R}^3)$ and the boundary condition ($\textsc{tms}')$ of \eqref{eq:DHab} are preserved. Concerning the former, we observe from \eqref{eq:Tlambda} that
\[
 (T_{\lambda}-T_{\lambda'})\xi\;=\;Q\xi + R\xi
\]
where
\[
 \begin{split}
  \widehat{(Q\xi)}(p)\;&:=\;2\pi^2\big(\sqrt{\nu p^2+\lambda}-\sqrt{\nu p^2+\lambda'} \big)\,\widehat{\xi}(p) \\
  &=\; (\lambda-\lambda')\,\frac{2\pi^2}{\sqrt{\nu p^2+\lambda}+\sqrt{\nu p^2+\lambda'}}\,\widehat{\xi}(p)\\
  \widehat{(R\xi)}(p)\;&:=\;(\lambda'-\lambda)\int_{\mathbb{R}^3}\frac{\widehat{\xi}(q)}{(p^2+q^2+\mu p\cdot q+\lambda)(p^2+q^2+\mu p\cdot q+\lambda')}\,\ud q\,.
 \end{split}
\]
Since $Q$ is the Fourier multiplier by a bounded function, one has $\|Q\xi\|_{H^{\frac{1}{2}}}\lesssim\|\xi\|_{H^{\frac{1}{2}}}$. Moreover, using H\"older's inequality, 
\begin{align*}
&\left\lVert R \xi \right\rVert^2_{H^{\frac{1}{2}}(\mathbb{R}^3)}\;=\;\big\|\langle p \rangle^{\frac{1}{2}} \widehat{R\xi}\,\big\|_2^2  \nonumber\\ 
&\leqslant \;\left\lVert \xi \right\rVert^2_{H^{\frac{1}{2}}}(\lambda-\lambda')^2\!\iint_{\mathbb R^6}\frac{\langle p \rangle }{\langle q \rangle(p^2+q^2+\mu pq+\lambda)^2(p^2+q^2+\mu pq+\lambda')^2}\,\ud q\, \ud p
\end{align*}
and the latter integral is finite. This shows that $(T_{\lambda}-T_{\lambda'})\xi\in H^{\frac{1}{2}}(\mathbb{R}^3)$, whence also $(T_{\lambda'}+\alpha\mathbbm{1})\xi\in H^{\frac{1}{2}}(\mathbb{R}^3)$.

Last, concerning the boundary condition $(\textsc{tms}')$, one has
\[
 \begin{split}
  \int_{\mathbb{R}^3}\widehat{F^{\lambda'}}(p_1,p_2)\,\ud p_2\;&=\;\int_{\mathbb{R}^3}\Big(\widehat{F^{\lambda}}(p_1,p_2)+\widehat{u_\xi^{\lambda}}(p_1,p_2)-\widehat{u_\xi^{\lambda'}}(p_1,p_2)\Big)\,\ud p_2 \\
  &=\;\widehat{(T_\lambda\xi)}(p_1)+\alpha\,\widehat{\xi}(p_1)+\big(\,\widehat{(T_{\lambda'}\xi)}(p_1)-\widehat{(T_\lambda\xi)}(p_1) \,\big) \\
   &=\;\widehat{(T_{\lambda'}\xi)}(p_1)+\alpha\,\widehat{\xi}(p_1)
 \end{split}
\]
where we applied $(\textsc{tms}')$ to $F^{\lambda}$ and the asymptotics \eqref{eq:TMS_pre-asymptotics} in the second identity. Therefore, $(\textsc{tms}')$ is preserved.

The conclusion is that the overall decomposition \eqref{eq:DHab} of $\mathcal{D}(H_\alpha)$ is invariant under the change of $\lambda$, for arbitrary $\lambda>0$.
\end{proof}

Let us prove now, through an analogous, yet independent argument, that the quadratic form too for $H_\alpha$ is given by \eqref{eq:DHab_form} for arbitrary shift parameter.

\begin{lemma}\label{lem:samedomainshift-2}
 Assume that the quadratic form $H_\alpha$ is only defined by \eqref{eq:DHab_form} for a given $\lambda>0$. Then the form $H_\alpha$ satisfies \eqref{eq:DHab_form} for any other $\lambda'>0$.
\end{lemma}

\begin{proof}
Let us decompose a generic $g\in\mathcal{D}[H_\alpha]$ as $g=F^{\lambda}+u^{\lambda}_\xi$ according to \eqref{eq:DHab_form}.
For another $\lambda'>0$ let
\[
 F^{\lambda'}\;:=\;F^{\lambda}+u^{\lambda}_\xi-u^{\lambda'}_\xi\,.
\]
Since $u^{\lambda}_\xi-u^{\lambda'}_\xi\in H^2_\mathrm{f}(\mathbb{R}^3\times\mathbb{R}^3)$, then $F^{\lambda'}\in H^1_\mathrm{f}(\mathbb{R}^3\times\mathbb{R}^3)$. Thus, the identity
\[\tag{i}
 g\;=\;F^{\lambda}+u^{\lambda}_\xi\;=\;F^{\lambda'}+u^{\lambda'}_\xi
\]
shows that the decomposition of $g$ into a regular part in $H^1_\mathrm{f}(\mathbb{R}^3\times\mathbb{R}^3)$ and a singular part of the type $u_\xi$ is valid also for $\lambda'$.

It remains to prove that the evaluation of $H_\alpha[g]$ is the same both with the $\lambda$-decomposition and with the $\lambda'$-decomposition of $g$. Explicitly, from \eqref{eq:DHab_form} one has
\[
 \begin{split}
  H_\alpha&[F^{\lambda'}+u^{\lambda'}_\xi]\;=\;H_\alpha[F^{\lambda}+u^{\lambda}_\xi] \\
  &=\;\lambda\|F^{\lambda'}\|_\cH^2-\lambda\|F^{\lambda'}+u_\xi^{\lambda'}\|_\cH^2+H_\mathrm{free}[F^{\lambda'}]+2\langle\xi,(T_{\lambda'}+\alpha\mathbbm{1}) \xi\rangle_{H^{\frac{1}{2}},H^{-\frac{1}{2}}}+\mathcal{K}_\alpha[g]\,,
 \end{split}
\]
where
\[\tag{ii}
 \begin{split}
  \mathcal{K}_\alpha[g]\;:=&\;\lambda\|F^{\lambda}\|_\cH^2-\lambda'\|F^{\lambda'}\|_{\cH}^2-\lambda\|F^{\lambda}+u_\xi^{\lambda}\|_\cH^2+\lambda'\|F^{\lambda'}+u_\xi^{\lambda'}\|_\cH^2 \\
  &\quad+H_\mathrm{free}[F^{\lambda}]-H_\mathrm{free}[F^{\lambda'}]+2\langle\xi,(T_{\lambda}-T_{\lambda'}) \xi\rangle_{H^{\frac{1}{2}},H^{-\frac{1}{2}}} \,,
 \end{split}
\]
and we need to prove that $\mathcal{K}_\alpha[g]= 0$.

We re-write the first line in the r.h.s.~above as
\[\tag{iii}
 \begin{split}
  &\lambda\|F^{\lambda}\|_\cH^2-\lambda'\|F^{\lambda'}\|_{\cH}^2-\lambda\|F^{\lambda}+u_\xi^{\lambda}\|_\cH^2+\lambda'\|F^{\lambda'}+u_\xi^{\lambda'}\|_\cH^2 \\
  &\quad =\;\langle F^{\lambda'},\lambda'u_\xi^{\lambda'}\rangle_{\cH}+\langle\lambda'u_\xi^{\lambda'},F^{\lambda'}\rangle_{\cH}+\lambda'\|u_\xi^{\lambda'}\|_{\cH}^2 \\
  &\qquad\quad-\langle F^{\lambda},\lambda u_\xi^{\lambda}\rangle_{\cH}-\langle\lambda u_\xi^{\lambda},F^{\lambda}\rangle_{\cH}-\lambda\|u_\xi^{\lambda}\|_{\cH}^2\,.
 \end{split}
\]
Moreover,
\[
 \begin{split}
  &H_\mathrm{free}[F^{\lambda}]-H_\mathrm{free}[F^{\lambda'}]\;=\;\langle H_{\mathrm{free}}^{\frac{1}{2}}F^\lambda,H_{\mathrm{free}}^{\frac{1}{2}}F^\lambda\rangle_{\cH}-\langle H_{\mathrm{free}}^{\frac{1}{2}}F^{\lambda'},H_{\mathrm{free}}^{\frac{1}{2}}F^{\lambda'}\rangle_{\cH} \\
  &\quad=\;\langle H_{\mathrm{free}}^{\frac{1}{2}}(F^\lambda-F^{\lambda'}),H_{\mathrm{free}}^{\frac{1}{2}}F^\lambda\rangle_{\cH}-\langle H_{\mathrm{free}}^{\frac{1}{2}}F^{\lambda'},H_{\mathrm{free}}^{\frac{1}{2}}(F^\lambda-F^{\lambda'})\rangle_{\cH}\,.
 \end{split}
\]
One has
\[
 \begin{split}
  \langle H_{\mathrm{free}}^{\frac{1}{2}}(F^\lambda-F^{\lambda'}),H_{\mathrm{free}}^{\frac{1}{2}}F^\lambda\rangle_{\cH}\;&=\;\langle H_{\mathrm{free}}^{\frac{1}{2}}(u_\xi^{\lambda'}-u_\xi^{\lambda}),H_{\mathrm{free}}^{\frac{1}{2}}F^\lambda\rangle_{\cH} \\
  &=\;\langle H_{\mathrm{free}}(u_\xi^{\lambda'}-u_\xi^{\lambda}),F^\lambda\rangle_{\cH} \\
  &=\;\langle \mathring{H}^*(u_\xi^{\lambda'}-u_\xi^{\lambda}),F^\lambda\rangle_{\cH} \\
  &=\;\langle(\lambda u_\xi^{\lambda}-\lambda' u_\xi^{\lambda'}),F^\lambda\rangle_{\cH}\,,
 \end{split}
\]
having used (i) for the first identity, the fact that $u^{\lambda}_\xi-u^{\lambda'}_\xi\in H^2_\mathrm{f}(\mathbb{R}^3\times\mathbb{R}^3)=\mathcal{D}(H_{\mathrm{free}})$ for the second identity, the inclusion $\mathring{H}\subset H_{\mathrm{free}}\subset \mathring{H}^*$ for the third identity, and the property \eqref{eq:kerHstar} for the last identity. Analogously,
\[
 \langle H_{\mathrm{free}}^{\frac{1}{2}}F^{\lambda'},H_{\mathrm{free}}^{\frac{1}{2}}(F^\lambda-F^{\lambda'})\rangle_{\cH}\;=\;\langle F^\lambda,(\lambda u_\xi^{\lambda}-\lambda' u_\xi^{\lambda'})\rangle_{\cH}\,,
\]
whence
\[\tag{iv}
 H_\mathrm{free}[F^{\lambda}]-H_\mathrm{free}[F^{\lambda'}]\;=\;\langle(\lambda u_\xi^{\lambda}-\lambda' u_\xi^{\lambda'}),F^\lambda\rangle_{\cH}-\langle F^\lambda,(\lambda u_\xi^{\lambda}-\lambda' u_\xi^{\lambda'})\rangle_{\cH}\,.
\]

Now, plugging (iii) and (iv) into (ii) yields
\[
\begin{split}
 \mathcal{K}_\alpha[g&]-2\langle\xi,(T_{\lambda}-T_{\lambda'}) \xi\rangle_{H^{\frac{1}{2}},H^{-\frac{1}{2}}}\;= \\
 &=\;\langle (F^{\lambda'}-F^{\lambda}),\lambda u_\xi^{\lambda}\rangle_{\cH}+\langle\lambda' u_\xi^{\lambda'},(F^{\lambda'}-F^{\lambda})\rangle_{\cH}-\lambda\|u_\xi^{\lambda}\|_{\cH}^2+\lambda'\|u_\xi^{\lambda'}\|_{\cH}^2 \\
 &=\;\langle (u_\xi^{\lambda}-u_\xi^{\lambda'}),\lambda u_\xi^{\lambda}\rangle_{\cH}+\langle\lambda' u_\xi^{\lambda'},(u_\xi^{\lambda}-u_\xi^{\lambda'})\rangle_{\cH}-\lambda\|u_\xi^{\lambda}\|_{\cH}^2+\lambda'\|u_\xi^{\lambda'}\|_{\cH}^2\,,
\end{split}
\]
having used again (i) in the last identity above, whence
\[\tag{v}
 \mathcal{K}_\alpha[g]\;=\;2\langle\xi,(T_{\lambda}-T_{\lambda'}) \xi\rangle_{H^{\frac{1}{2}},H^{-\frac{1}{2}}}-(\lambda-\lambda')\langle u_\xi^{\lambda'},u_{\xi}^\lambda\rangle_{\cH}\,.
\]

Next, for the vanishing of (v), using the computations made in the proof of Lemma \ref{lem:samedomainshift-1}, we have
\[\tag{vi}
 \begin{split}
  \langle\xi,&(T_{\lambda}-T_{\lambda'}) \xi\rangle_{H^{\frac{1}{2}},H^{-\frac{1}{2}}}\;=\;(\lambda-\lambda')\int_{\mathbb{R}^3}\ud p\,\overline{\widehat{\xi}(p)}\Big(\frac{2\pi^2\,\widehat{\xi}(p)}{\sqrt{\nu p^2+\lambda}+\sqrt{\nu p^2+\lambda'}}\,- \\
  &\qquad\qquad\qquad\qquad -\int_{\mathbb{R}^3}\ud q\,\frac{\widehat{\xi}(q)}{(p^2+q^2+\mu p\cdot q+\lambda)(p^2+q^2+\mu p\cdot q+\lambda')}\Big)\,.
 \end{split}
\]
On the other hand, through analogous computations, one has
\begin{equation*} \tag{vii}
\begin{split}
\langle u_{\xi}^{\lambda'}&, u_{\xi}^{\lambda}\rangle_\cH\;=\;\iint_{\mathbb{R}^3\times\mathbb{R}^3}{\frac{\overline{\,\widehat{\xi}(p)}-\overline{\,\widehat{\xi}(q)}}{p^2+q^2+\mu p\cdot q+\lambda'}\:\frac{\widehat \xi (p)-\widehat \xi (q)}{p^2+q^2+\mu p\cdot q+\lambda}\,\ud p\,\ud q} \\
&=\;2\iint_{\mathbb{R}^3\times\mathbb{R}^3}{\frac{\overline{\,\widehat{\xi}(p)}\,\widehat\xi(p)-\overline{\,\widehat{\xi}(p)}\,\widehat\xi(q)}{(p^2+q^2+\mu p\cdot q+\lambda)(p^2+q^2+\mu p\cdot q+\lambda')}\,\ud p \,\ud q} \\
&=\;2\int_{\mathbb{R}^3}\ud p\,\overline{\widehat{\xi}(p)}\Big(\frac{2\pi^2\,\widehat{\xi}(p)}{\sqrt{\nu p^2+\lambda}+\sqrt{\nu p^2+\lambda'}}\,- \\
  &\qquad\qquad\qquad -\int_{\mathbb{R}^3}\ud q\,\frac{\widehat{\xi}(q)}{(p^2+q^2+\mu p\cdot q+\lambda)(p^2+q^2+\mu p\cdot q+\lambda')}\Big),
\end{split}
\end{equation*}
where we used the symmetry under exchange $p\leftrightarrow q$ in the second step and we computed an explicit integration in $\ud q$ in the last step. Plugging (vi) and (vii) into (v) proves indeed that $\mathcal{K}_{\alpha}[g]=0$, which completes the proof.
\end{proof}

\section{Improved lower bound for $H_\alpha$}\label{sec:improved_lower_bound}

In this Section we show how, based on the findings of \cite{CDFMT-2012}, the lower bound determined therein for the operator $H_\alpha$ when $\alpha<0$, that is, the estimate \eqref{eq:old_lower_bound}, can be improved to the actual lower bound reported in Theorem \ref{thm:properties_of_Halpha}(ii), estimate \eqref{eq:inf_spec_Halpha}.

In either case the lower bound is far from optimal (and diverges when $m\downarrow m^*$), since it is obtained by  merely requiring that in the expression \eqref{eq:DHab_form} for the quadratic form $H_{\alpha}[F^\lambda+u_\xi^\lambda]$ the auxiliary parameter $\lambda$ is chosen sufficiently large, say, $\lambda\geqslant\lambda_\natural$, so as to guarantee the \emph{positivity of the form of the charges},
\begin{equation}\label{eq:lower_bound_charge_form}
 \int_{\mathbb{R}^3}\overline{\,\widehat{\xi}(p)}\,\big(\widehat{(T_\lambda\xi)}(p)+\alpha\,\widehat{\xi}(p)\big)\,\ud p\;\geqslant\;0\qquad\qquad (\lambda\geqslant\lambda_\natural)\,.
\end{equation}
When this is the case, \eqref{eq:DHab_form} and \eqref{eq:lower_bound_charge_form} imply at once the lower bound
\begin{equation}\label{eq:cons_nat}
 H_{\alpha}[F^\lambda+u_\xi^\lambda]\;\geqslant\;-\lambda_\natural\,\|F^\lambda+u_\xi^\lambda\|_\cH^2\,.
\end{equation}

The detailed analysis made in \cite{CDFMT-2012} of $T_\lambda$, as an operator on $L^2(\mathbb{R}^3)$, shows that
\begin{equation}\label{eq:chargesform_temp_est}
 \begin{split}
 \int_{\mathbb{R}^3}\overline{\,\widehat{\xi}(p)}\,&\widehat{(T_\lambda\xi)}(p)\,\ud p\;\geqslant\;2\pi^2(1-\Lambda(m))\!\!\int_{\mathbb{R}^3}\sqrt{{\textstyle\frac{m}{m+1}}p^2+\lambda}\;|\widehat{\xi}(p)|^2\,\ud p \\
 &\geqslant\;2\pi^2\!\!\int_{\mathbb{R}^3}\sqrt{\nu\,p^2+\lambda}\;|\widehat{\xi}(p)|^2\,\ud p-2\pi^2\Lambda(m)\sqrt{\nu}\!\int_{\mathbb{R}^3}|p|\,|\widehat{\xi}(p)|^2\,\ud p
 \end{split}
 \end{equation}
(the first inequality above comes from \cite[Eq.~(3.51)]{CDFMT-2012}). Considering only the first inequality in \eqref{eq:chargesform_temp_est}, and requiring that $2\pi^2(1-\Lambda(m))\sqrt{\lambda}\geqslant|\alpha|$, implies 
\eqref{eq:lower_bound_charge_form}-\eqref{eq:cons_nat} precisely in the form of the lower bound  \eqref{eq:old_lower_bound} available in the literature.

Instead, considering the second inequality in \eqref{eq:chargesform_temp_est}, one has
\[
  \int_{\mathbb{R}^3}\overline{\,\widehat{\xi}(p)}\,\big(\widehat{(T_\lambda\xi)}(p)+\alpha\,\widehat{\xi}(p)\big)\,\ud p\;\geqslant\;\int_{\mathbb{R}^3}f(|p|)\,|\widehat{\xi}(p)|^2\,\ud p
\]
with
\[
 f(\rho)\;:=\;2\pi^2\sqrt{\nu\rho^2+\lambda}-2\pi^2\rho\,\Lambda(m)\sqrt{\nu}-|\alpha|\,.
\]
It is easy to see that $f(\rho)$ attains its minimum at $\rho=\rho_0:=\sqrt{\frac{\lambda}{\nu}}\,\frac{\Lambda(m)}{\sqrt{1-\Lambda(m)^2}}$ and 
\[
 f(\rho)\;\geqslant\;f(\rho_0)\;=\;2\pi^2\sqrt{\lambda}\,\sqrt{1-\Lambda(m)^2}-|\alpha|\,.
\]
Thus, $f(\rho_0)\geqslant 0$ for $\lambda\geqslant \lambda_\natural:=\frac{\alpha^2}{4\pi^4}\,\frac{1}{1-\Lambda(m)^2}$. This implies \eqref{eq:lower_bound_charge_form}-\eqref{eq:cons_nat} for such $\lambda_\natural$, whence indeed the ameliorated lower bound \eqref{eq:inf_spec_Halpha}.

\section{Angular decomposition}\label{sec:angular_decomp}

The preparatory material of this Section is an adaptation of the analysis developed by one of us in \cite{CDFMT-2012} in collaboration with Dell'Antonio, Correggi, Finco, and Teta, and concerns the reduction of the operator $T_\lambda$ with respect to the canonical decomposition
 \begin{equation}\label{eq:can_decomp}
  L^2(\mathbb{R}^3)\;\cong\;\bigoplus_{\ell=0}^\infty L^2_{\ell}(\mathbb{R}^3)\,,\quad L^2_{\ell}(\mathbb{R}^3)\,\cong\,L^2(\mathbb{R}^+,r^2\,\ud r)\otimes \mathrm{span}\{Y_{\ell,-\ell},\dots, Y_{\ell,\ell}\}
 \end{equation}
 into subspaces of definite angular symmetry (here the $Y_{\ell,M}$'s are the spherical harmonics on $\mathbb{S}^2$). The useful notation
\begin{equation}
 H^s_+(\mathbb{R}^3)\;:=\;\bigoplus_{\mathrm{even}\:\ell}H_\ell^{s}(\mathbb{R}^3)\,,\qquad H^s_-(\mathbb{R}^3)\;:=\;\bigoplus_{\mathrm{odd}\:\ell}H_\ell^{s}(\mathbb{R}^3)
\end{equation}
will be also used in the following. We refer to \cite[Sec.~3]{CDFMT-2012} for the details.

A generic $\xi\in H^{\frac{1}{2}}(\mathbb{R}^3)$ is decomposed with respect to \eqref{eq:can_decomp} and with polar coordinates $p\equiv|p|\Omega_p\in\mathbb{R}^3$ as
\begin{equation}\label{eq:xi-decomp}
 \widehat{\xi}(p)\;=\;\sum_{\ell=0}^\infty\sum_{M=-\ell}^\ell f_{\ell,M}(|p|)\, Y_{\ell,M}(\Omega_p)
\end{equation}
for some $f_{\ell,M}\in L^2(\mathbb{R}^+,r^2\sqrt{r^2+1}\,\ud r)$. The operator $T_\lambda$ commutes with the angular momentum operator, therefore its expectation on $\xi$ can be decomposed as
\begin{equation}\label{eq:Tform_decomp_1}
 \langle\xi,T_\lambda \xi\rangle_{H^{\frac{1}{2}},H^{-\frac{1}{2}}}\;=\;\sum_{\ell=0}^\infty\sum_{M=-\ell}^\ell\Big(\Phi_\lambda[f_{\ell,M}]+\Psi_{\lambda,\ell}[f_{\ell,M}]\Big)\,,
\end{equation}
where
\begin{equation}\label{eq:Tform_decomp_2}
\begin{split}
 \Phi_\lambda[f]\;&:=\;2\pi^2\!\int_0^{+\infty}\!\!\ud r\,r^2\sqrt{\nu r^2+\lambda}\,|f(r)|^2 \\
 \Psi_{\lambda,\ell}[f]\;&:=\;2\pi\!\int_0^{+\infty}\!\!\ud r\int_0^{+\infty}\!\!\ud r'\,r^2\,\overline{f(r)}\,r'^2f(r')\int_{-1}^1\ud y\,\frac{P_\ell(y)}{r^2+r'^2+\mu r r' y + \lambda}
\end{split}
\end{equation}
and
\begin{equation}\label{eq:LegendrePoly}
  P_\ell(y)\;=\;\frac{1}{2^\ell \ell!}\frac{\ud^\ell}{\ud y^\ell}(y^2-1)^{\ell}
 \end{equation}
 is the $\ell$-th Legendre polynomial, as follows straightforwardly from the addition formula \cite[Eq.~(8.814)]{Gradshteyn-tables-of-integrals-etc}
 \begin{equation}\label{eq:addition_formula}
 \begin{split}
  &\frac{1}{p^2+q^2+\mu p\cdot q+\lambda}\;= \\
  &\qquad\qquad=\;\sum_{\ell=0}^\infty 2\pi\int_{-1}^1\ud y\,\frac{P_\ell(y)}{p^2+q^2+\mu |p| |q| y+\lambda}\sum_{M=-\ell}^\ell \overline{Y_{\ell,M}(\Omega_q)}\,Y_{\ell,M}(\Omega_p)\,.
 \end{split}
 \end{equation}

 .

Simple manipulations \cite[Lemma 3.2]{CDFMT-2012} show that
\begin{equation}\label{eq:offdiagposeven}
   \Psi_{\lambda,\ell}[f]\;\geqslant\; 0\qquad\textrm{for even $\ell$}\,,
\end{equation}
whereas
\begin{equation}\label{eq:offdiagint}
\begin{split}
 \Psi_{\lambda,\ell}[f]\;=&\;-\sum_{\substack{j=\ell \\ j\,\textrm{odd}}}^\infty\mu^j B_{\ell,j}\int_0^{+\infty}\!\!\ud s\,s^j e^{-\lambda s}\Big|\int_0^{+\infty}\!\!\ud r f(r) \,r^{2+j}\,e^{-sr^2} \Big|^2\,, \\
 B_{\ell,j}\;:=&\;\frac{2\pi}{\,2^\ell \,\ell!\, j!\,}\int_{-1}^1\ud y\,(1-y^2)^\ell\,\frac{\ud^\ell}{\ud y^\ell}\,y^j\;\geqslant\;0\qquad\textrm{for odd $\ell$ and $j$}\,,
\end{split}
\end{equation}
whence
 \begin{equation}\label{eq:psi0psi1oldineq}
  \Psi_{0,\ell}[f]\;\leqslant\;\Psi_{\lambda,\ell}[f]\;\leqslant\;0\qquad\textrm{for odd $\ell$}\,.
 \end{equation}

Furthermore, by means of a change of variable and the Fourier transform \cite[Lemma 3.2]{CDFMT-2012}, one can re-write
\begin{equation}\label{eq:Psi0repr-1}
 \Psi_{0,\ell}[f]\;=\;-2\pi^2\int_\mathbb{R}\sigma_\ell^{(m)}(k)\,|f^\sharp(k)|^2\qquad\textrm{for odd $\ell$}\,,
\end{equation}
where
\begin{equation}\label{eq:Psi0repr-2}
 f^\sharp(k)\;:=\;\frac{1}{\sqrt{2\pi}}\int_\mathbb{R}\ud x\,e^{-\ii k x}e^{2x}f(e^x)
\end{equation}
and
\begin{equation}\label{eq:Psi0repr-3}
 \sigma_\ell^{(m)}(k)\;:=\;\frac{1}{2}\int_{-1}^1\!\ud y\,P_\ell(y)\,\frac{\sinh\big(k\arcsin\frac{y}{m+1}\big)}{\,\cos\big(\arcsin\frac{y}{m+1}\big)\sinh(\frac{k\pi}{2})\,}\,.
\end{equation}
The map $k\mapsto \sigma_\ell^{(m)}(k)$ is smooth, strictly positive, even, vanishing as $k\to\pm\infty$, and monotone decreasing for $k>0$. Moreover,
\begin{equation}\label{eq:sigmaellk_tower}
 \sigma_1^{(m)}(k)\;>\;\sigma_3^{(m)}(k)\;>\;\sigma_5^{(m)}(k)\;>\;\cdots\;>\;0
\end{equation}
and 
\begin{equation}
 \sigma_\ell^{(m_1)}(k)\;>\;\sigma_\ell^{(m_2)}(k)\qquad\textrm{if } m_1<m_2\,.
\end{equation}

Let us conclude this Section by highlighting a useful consequence of the above analysis.

\begin{lemma}[Monotonicity of $T_\lambda$ in $m$]\label{lem:T_monotonicity}
Let $m>m^*$ and $\lambda>0$. For each $\xi\in H^{1/2}_-(\mathbb{R}^3)$ the map
\[
 m\;\longmapsto\;\langle\xi,T_\lambda\xi\rangle_{H^{\frac{1}{2}},H^{-\frac{1}{2}}}
\]
is continuous and strictly monotone increasing.
\end{lemma}

\begin{proof}
 Owing to \eqref{eq:Tform_decomp_1}-\eqref{eq:Tform_decomp_2} the statement to prove is equivalent to the continuity and the strictly increasing monotonicity of the map 
 \[
  m\;\longmapsto\;\Phi_\lambda[f]+\Psi_{\lambda,\ell}[f]
 \]
 for every $\ell\in\mathbb{N}_0$ and $f\in L^2(\mathbb{R}^+,r^2\sqrt{r^2+1}\,\ud r)$. Now, the statement for $m\mapsto \Phi_\lambda[f]$ is obvious from the first of \eqref{eq:Tform_decomp_2}, and the statement for $m\mapsto \Psi_{\lambda,\ell}[f]$ follows at once from the representation \eqref{eq:offdiagint}.
 
\end{proof}

\section{Reduction lemmas}\label{eq:reduction_lemma}

In this Section we connect the eigenvalue problem for $H_\alpha$, when $\alpha<0$, with the eigenvalue problem for $T_\lambda$ or also, up to a re-scaling, for $T_1$. Next, we show that the charges of the eigenfunctions of $H_\alpha$ are of special angular symmetry.

The first main result is the following.

\begin{lemma}\label{lem:reduction_lemma}
 Let $m>m^*$, $\alpha<0$, and 
 \begin{equation}\label{eq:amiss_window}
  \lambda\;\in\;\Big(\frac{\alpha^2}{4\pi^4},\frac{\alpha^2}{4\pi^4(1-\Lambda(m)^2)}\Big]\,.
 \end{equation}
The following conditions are equivalent.
 \begin{itemize}
  \item[(i)] $g$ is an eigenfunction of $H_\alpha$ with eigenvalue $-\lambda$.
  \item[(ii)] $g=u_\xi^{\lambda}$, where $\xi\in H^{\frac{1}{2}}(\mathbb{R}^3)$ is an eigenfunction of $T_\lambda$ with eigenvalue $-\alpha$. 
  \item[(iii)] $g=u_\xi^{\lambda}$, where $\widehat{\xi}(p)=\widehat{\widetilde{\xi}}(\frac{p}{\sqrt{\lambda}})$ and $\widetilde{\xi}\in H^{\frac{1}{2}}(\mathbb{R}^3)$ is an eigenfunction of $T_1$ with eigenvalue $\frac{|\alpha|}{\sqrt{\lambda}}$.
 \end{itemize}
\end{lemma}

The idea behind Lemma \ref{lem:reduction_lemma} is an old one. That the eigenvalue problem for the three-body Hamiltonian with contact interaction could be reformulated as the eigenvalue problem for what we call now the charge operator is a property that was noticed first by Minlos and Faddeev \cite{Minlos-Faddeev-1961-1,Minlos-Faddeev-1961-2}. By means of a self-adjoint extension scheme \`a la Kre{\u\i}n-Vi\v{s}ik-Birman they could reproduce the celebrated Ter-Martirosyan--Skornyakov equation, that had been identified heuristically a few years earlier by Ter-Martirosyan and Skornyakov in the study of the bound states for the three-body system with zero-range interaction \cite{TMS-1956}. In the present (2+1 fermionic) case, the Ter-Martirosyan--Skornyakov equation has precisely the form 
\begin{equation}
 \alpha\,\xi + T_\lambda\xi\;=\;0
\end{equation}
in the unknown $(\xi,\lambda)$.
This equation, even if not derived rigorously, has been the object of intensive numerical investigation in the physical literature (see, e.g., \cite{Endo-Naidon-Ueda-2011}). 
On the mathematical side, the equivalence (i)$\Leftrightarrow$(ii) of Lemma \ref{lem:reduction_lemma} based upon the characterisation of $H_\alpha$ as a quadratic form was shown in \cite[Sec.~5]{DFT-1994} and more recently in  \cite[Sec.~III]{michelangeli-schmidbauer-2013}.


\begin{proof}[Proof of Lemma \ref{lem:reduction_lemma}] Let $g=F^\lambda+u_\xi^\lambda\in\mathcal{D}(H_\alpha)$. The range for possible eigenvalues $-\lambda$ of $H_\alpha$ is precisely the one indicated in the statement, as follows from \eqref{eq:where_disc_spec}.

(i)$\Leftrightarrow$(ii). The identity $H_\alpha g=-\lambda g$ is equivalent to $F^\lambda\equiv 0$, owing to \eqref{eq:action_DHab}, which in turn is equivalent to $T_\lambda\xi=-\alpha\xi$, owing to the condition $(\textsc{tms}')$ in \eqref{eq:DHab}. 

(ii)$\Leftrightarrow$(iii). Let $\widetilde{p}:=p/\sqrt{\lambda}$. One has
\[
 \begin{split}
&\widehat{(T_\lambda\xi)}(p)+\alpha\,\widehat{\xi}(p)\;=\;(2\pi^2\sqrt{\nu p^2+\lambda}+\alpha)\,\widehat{\xi}(p)+\int_{\mathbb{R}^3}\frac{\widehat{\xi}(q)}{p^2+q^2+\mu p\cdot q+\lambda}\,\ud q  \\
&=\;\sqrt{\lambda}\,(2\pi^2\sqrt{\nu \widetilde{p}^{\,2}+1}+{\textstyle\frac{\alpha}{\sqrt{\lambda}}})\,\widehat{\widetilde{\xi}}(\widetilde{p})+\sqrt{\lambda}\int_{\mathbb{R}^3}\frac{\widehat{\widetilde \xi}(q)}{\widetilde{p}^{\,2}+q^2+\mu \widetilde{p}\cdot q+1}\,\ud q \\
&=\;\sqrt{\lambda}\Big((\widehat{T_1\widetilde{\xi}})(\widetilde{p})+{\textstyle\frac{\alpha}{\sqrt{\lambda}}}\,\widehat{\widetilde{\xi}}(\widetilde{p}) \Big)\,,
 \end{split}
\]
which completes the proof.
\end{proof}

The same scaling argument used in the proof above allows one to conclude the following.

\begin{lemma}\label{lem:scaling}
Let $m>m^*$ and $\alpha<0$.
 For a constant $\lambda>0$ and for two functions $\xi,\widetilde{\xi}\in H^{\frac{1}{2}}(\mathbb{R}^3)$ such that $\widehat{\widetilde{\xi}}(p)=\widehat{\xi}(p\sqrt{\lambda})$, $p\in\mathbb{R}^3$, the following conditions are equivalent.
 \begin{itemize}
  \item[(i)] $\langle\xi,(T_\lambda+\alpha\mathbbm{1})\xi\rangle_{H^{\frac{1}{2}},H^{-\frac{1}{2}}}\;=\;0$ 
  \item[(ii)] $\langle\widetilde{\xi},T_1\widetilde{\xi}\rangle_{H^{\frac{1}{2}},H^{-\frac{1}{2}}}\;=\;\varepsilon\|\widetilde{\xi}\|_{L^2(\mathbb{R}^3)}^2$ for $\varepsilon:=\frac{|\alpha|}{\sqrt{\lambda}}$\,.
 \end{itemize}
\end{lemma}

Next, we show that the charge $\xi$ of an eigenfunction for $H_\alpha$, with $\alpha<0$, must have odd angular symmetry, i.e., $\ell\in\{1,3,5,\dots\}$, with respect to the canonical decomposition \eqref{eq:can_decomp}.

\begin{lemma}\label{lem:only_odd_symm}
 Let $m>m^*$ and $\alpha<0$, and let $g=u_\xi^\lambda$ be an eigenfunction of $H_\alpha$ with eigenvalue $-\lambda$ in the admissible spectral window \eqref{eq:amiss_window}, as prescribed by Lemma \ref{lem:reduction_lemma}. Then 
 \[
  \xi\;\in\;H^{1/2}_-(\mathbb{R}^3)\;\equiv\;\bigoplus_{\mathrm{odd}\; \ell} H^{1/2}_{\ell}(\mathbb{R}^3)\,.
 \]
\end{lemma}

\begin{proof}
Let us decompose $\xi$, and hence also $\widetilde{\xi}$, as in \eqref{eq:xi-decomp}, and then let us express the expectation $\langle\widetilde{\xi},T_1\widetilde{\xi}\rangle_{H^{\frac{1}{2}},H^{-\frac{1}{2}}}$ through the decomposition \eqref{eq:Tform_decomp_1}-\eqref{eq:Tform_decomp_2}.
If $\xi\in H^{1/2}_+(\mathbb{R}^3)\equiv\bigoplus_{\mathrm{even}\: \ell} H^{1/2}_{\ell}(\mathbb{R}^3)$, then \eqref{eq:Tform_decomp_2} and \eqref{eq:offdiagposeven} imply
\[
 \begin{split}
  \langle\widetilde{\xi},T_1\widetilde{\xi}\rangle_{L^2(\mathbb{R}^3)}\;&=\;\sum_{\substack{\ell,M \\ \ell\:\mathrm{even}}}\Big(\Phi_1[f_{\ell,M}]+\Psi_{1,\ell}[f_{\ell,M}]\Big)\;\geqslant\;\sum_{\substack{\ell,M \\ \ell\:\mathrm{even}}}\Phi_1[f_{\ell,M}] \\
  &\geqslant\;2\pi^2\sum_{\substack{\ell,M \\ \ell\:\mathrm{even}}}\|f_{\ell,M}\|^2_{L^2(\mathbb{R}^+,r^2\ud r)}\;=\;2\pi^2\|\widetilde{\xi}\|_{L^2(\mathbb{R}^3)}^2\,,
 \end{split}
\]
which prevents the equation $T_1\widetilde{\xi}=\varepsilon\widetilde{\xi}$ to hold with eigenvalue $\varepsilon<2\pi^2$. Owing to Lemma \ref{lem:reduction_lemma}, such $\xi$ cannot be the charge of an eigenfunction of $H_\alpha$, so the alternative is necessarily the one stated in the thesis.
\end{proof}

\section{Essential spectrum }\label{sec:ess_spec}

In this Section we start the proof Theorem \ref{thm:ess-spec}, which will be completed in the end of Section \ref{sec:finite_disc_spec}. In particular, here we prove that all the real numbers above an explicit $\alpha$-dependent threshold belong to  $\sigma_{\mathrm{ess}}(H_\alpha)$, and in the end of Section \ref{sec:finite_disc_spec} we show that such threshold is precisely the bottom of $\sigma_{\mathrm{ess}}(H_\alpha)$.

For the latter step we need first to prove the finiteness of the discrete spectrum below the considered $\alpha$-dependent threshold, a result that we shall establish in Section \ref{sec:finite_disc_spec}.
For the former, in order to show that a given real number $\mu$ belongs to $\sigma_{\mathrm{ess}}(H_\alpha)$, 
 we shall produce a \emph{singular sequence} \cite[Sec.~8.4]{schmu_unbdd_sa} for $H_\alpha$ at the value $\mu$.

In fact, the case $\alpha\geqslant 0$ is straightforward as compared to the case $\alpha<0$, let us discuss it first.

\begin{proposition}\label{prop:ess-spec-alphapositive}
 Let $m>m^*$ and $\alpha\geqslant 0$. Then $\sigma_{\mathrm{ess}}(H_\alpha)=[0,+\infty)$.
\end{proposition}

\begin{proof}
As well known, $\sigma_{\mathrm{ess}}(H_{\mathrm{free}})=[0,+\infty)$ and for an arbitrary $\mu\geqslant 0$ it is always possible to find a sequence $(F_n)_n$ in $H^2_\mathrm{f}(\mathbb{R}^3\times\mathbb{R}^3)$ that is a singular sequence for $H_{\mathrm{free}}$ at $\mu$, i.e.,
\[
 \liminf_n\|F_n\|_{\cH}>0\,,\qquad F_n\rightharpoonup 0\,,\qquad \|(H_{\mathrm{free}}-\mu\mathbbm{1})F_n\|_{\cH}\to 0\,,
\]
and such that additionally $\int_{\mathbb{R}^3}\widehat{F}_n(p,q)\ud p=0$.
The ``vanishing condition'' for $F_n$ at the coincidence hyperplanes is clearly non-restrictive for the construction of a singular sequence for $H_{\mathrm{free}}$, yet it guarantees that $F_n\in\mathcal{D}(H_\alpha)$, owing to \eqref{eq:DHab}. Thus, for arbitrary $\lambda>0$,
\[
 \begin{split}
  (H_\alpha-\mu\mathbbm{1})F_n\;&=\;(H_\alpha+\lambda\mathbbm{1})F_n-(\mu+\lambda)F_n \\
  &=\;(H_{\mathrm{free}}+\lambda\mathbbm{1})F_n-(\mu+\lambda)F_n \;=\;(H_{\mathrm{free}}-\mu\mathbbm{1})F_n\,,
 \end{split}
\]
having used \eqref{eq:action_DHab} in the second step. Therefore $\|(H_{\alpha}-\mu\mathbbm{1})F_n\|_{\cH}\to 0$ and $(F_n)_n$ is also a singular sequence for $H_{\alpha}$ at $\mu$, i.e., $\mu\in\sigma_{\mathrm{ess}}(H_\alpha)$. We proved that $[0,+\infty)\subset\sigma_{\mathrm{ess}}(H_\alpha)$; on the other hand we know from Theorem \ref{thm:properties_of_Halpha}, formula \eqref{eq:inf_spec_Halpha}, that $\sigma(H_\alpha)\subset[0,+\infty)$. The conclusion is then $\sigma_{\mathrm{ess}}(H_\alpha)=\sigma(H_\alpha)=[0,+\infty)$.
\end{proof}

The case $\alpha<0$ is more subtle. The counterpart to Proposition \ref{prop:ess-spec-alphapositive} that we shall prove in this Section is the following.

\begin{proposition}\label{prop:ess-spec-alphanegative_1}
 Let $m>m^*$ and $\alpha< 0$. Then $\sigma_{\mathrm{ess}}(H_\alpha)\supset [-\frac{\alpha^2}{4\pi^4},+\infty)$.
\end{proposition}

For the proof of Proposition \ref{prop:ess-spec-alphanegative_1} we need to recall from the previous literature a few extra technical facts concerning the operator $H_\alpha$ which we did not include in our outlook in Section \ref{sec:setting_and_main_results}, and in fact are at the basis of the very construction of $H_\alpha$ as a self-adjoint extension of $\mathring{H}$ within the Kre{\u\i}n-Vi\v{s}ik-Birman extension scheme.

\begin{proposition}[Space of charges $H^{-1/2}_{W_\lambda}$]\label{prop:space_of_charges}
 Let $m>0$ and $\lambda>0$.
 \begin{itemize}
  \item[(i)] The linear map $W_\lambda:H^{-\frac{1}{2}}(\mathbb{R}^3)\to H^{\frac{1}{2}}(\mathbb{R}^3)$ defined by
  \begin{equation}\label{eq:Wlambda}
\widehat{(W_\lambda\,\xi)}(p)\;:=\;\frac{2\pi^2}{\sqrt{\nu p^2+\lambda}\,}\,\widehat{\xi}(p)-2\!\int_{\mathbb{R}^3}\frac{\widehat{\xi}(q)}{(p^2+q^2+\mu p\cdot q+\lambda)^2}\,\ud q
\end{equation}
is bounded, positive, invertible, and onto.
 \item[(ii)] For generic $u_\xi^\lambda,u_\eta^\lambda\in\ker (\mathring{H}^*+\lambda\mathbbm{1})$ one has
 \begin{equation}\label{eq:scalar_products}
 \langle u_\xi^\lambda,u_\eta^\lambda\rangle_{\cH}\;=\;\langle \xi,W_\lambda\eta\rangle_{H^{-\frac{1}{2}},H^{\frac{1}{2}}}
 \end{equation}
 and the above expression defines a scalar product in $H^{-\frac{1}{2}}(\mathbb{R}^3)$ which is equivalent to the standard $H^{-\frac{1}{2}}$-scalar product.
 \item[(iii)] Denoting by $H^{-1/2}_{W_\lambda}$ the space of the $H^{-\frac{1}{2}}$-functions equipped with the scalar product \eqref{eq:scalar_products}, the map 
 \begin{equation}\label{eq:isomorphism_Ulambda}
\begin{split}
U_\lambda\,:\,\ker (\mathring{H}^*+\lambda\mathbbm{1})\;&\;\xrightarrow[]{\;\;\;\cong\;\;\;}\;H^{-1/2}_{W_\lambda}(\mathbb{R}^3)\,,\qquad u_\xi^\lambda \longmapsto \,\xi
\end{split}
\end{equation}
is an isomorphism between Hilbert spaces, with $\ker (\mathring{H}^*+\lambda\mathbbm{1})$ equipped with  the standard scalar product inherited from $\cH$. 
 \end{itemize}
\end{proposition}

The proof of Proposition \ref{prop:space_of_charges} is presented in \cite[Sec.~2]{Minlos-2011-preprint_May_2010}, \cite[Sec.~4.2]{MO-2016}, and \cite[Sec.~3]{MO-2017}.

\begin{proposition}[Birman parametrisation of $H_\alpha$]\label{prop:Birman_parametr_Halpha}
 Let $m>m^*$ and $\lambda>0$.
 \begin{itemize}
  \item[(i)] The linear operator $\mathcal{A}_{\lambda,\alpha}$ defined by
  \begin{equation}\label{eq:DA_op}
\begin{split}
\mathcal{D}(\mathcal{A}_{\lambda,\alpha})\;&:=\;\{\xi\in H^{\frac{1}{2}}(\mathbb{R}^3)\,|\,(T_\lambda+\alpha\mathbbm{1})\xi\in  H^{\frac{1}{2}}(\mathbb{R}^3)\} \\
\mathcal{A}_{\lambda,\alpha}\,\xi\;&:=\;2\,W_\lambda^{-1}(T_\lambda+\alpha\mathbbm{1})\,\xi
\end{split}
\end{equation}
 is self-adjoint on the Hilbert space $H^{-1/2}_{W_\lambda}$ introduced in Prop.~\ref{prop:space_of_charges}(iii).
 \item[(ii)] One has
 \begin{equation}\label{eq:Birman_parametr_Halpha}
  \mathcal{D}(H_\alpha)\;=\;\left\{ g=f+(H_\mathrm{free}+\lambda\mathbbm{1})^{-1} u^\lambda_{\mathcal{A}_{\lambda,\alpha}\xi}+u^\lambda_\xi\left|
  \!\begin{array}{c}
   f\in\mathcal{D}(\mathring{H}) \\
   \xi\in\mathcal{D}(\mathcal{A}_{\lambda,\alpha})
  \end{array}\!\!\!\right.\right\}.
 \end{equation}
 \end{itemize}
\end{proposition}

The proof of Proposition \ref{prop:space_of_charges} is presented in \cite[Sec.~4.3]{MO-2016} and \cite[Sections 3-6]{MO-2017}. In the notation of formulas \eqref{eq:DHab} and \eqref{eq:Birman_parametr_Halpha}, the regular $H^2$-part of $g$ is the function $F^\lambda=f+(H_\mathrm{free}+\lambda\mathbbm{1})^{-1} u^\lambda_{\eta}$ with charge $\eta=\mathcal{A}_{\lambda,\alpha}\xi$, and in fact the condition $\eta=\mathcal{A}_{\lambda,\alpha}\xi$ is tantamount as the Ter-Martirosyan-Skornyakov condition $(\textsc{tms}')$. Moreover, from \eqref{eq:kerHstar} and \eqref{eq:Birman_parametr_Halpha}, and consistently with \eqref{eq:action_DHab}, 
\begin{equation}\label{eq:actionHalpha2}
 (H_\alpha+\lambda\mathbbm{1}) g\;=\;(\mathring{H}+\lambda\mathbbm{1})f+u^\lambda_{\mathcal{A}_{\lambda,\alpha}\xi}\,.
\end{equation}

\begin{proposition}[Additional regularity properties of $T_\lambda$]\label{prop:regpropTlambda}
Let $m>0$ and $\lambda>0$. Then
\begin{equation}\label{eq:T-ell-3/2-1/2_excluded}
\|T_\lambda\xi\|_{H^{s-1}}\;\lesssim\;\|\xi\|_{H^{s}}\qquad s\in\,\textstyle{(-\frac{1}{2},\frac{3}{2})}
\end{equation}
and in addition
\begin{equation}\label{eq:T-ell-3/2-1/2_included}
\|T_\lambda\xi\|_{H^{s-1}}\;\lesssim\;\|\xi\|_{H^{s}}\qquad\forall\xi\in H^{s}_\ell(\mathbb{R}^3) \qquad 
\begin{cases}
\;s\in\,\textstyle{[-\frac{1}{2},\frac{3}{2}]} \\
\quad \ell\geqslant 1\,,
\end{cases}
\end{equation}
whereas $T_\lambda$ fails to map continuously $H^{3/2}_{\ell=0}(\mathbb{R}^3)$ into $H^{1/2}_{\ell=0}(\mathbb{R}^3)$ or $H^{-1/2}_{\ell=0}(\mathbb{R}^3)$ into $H^{-3/2}_{\ell=0}(\mathbb{R}^3)$.
\end{proposition}

The proof of Proposition \ref{prop:regpropTlambda} is presented in \cite[Prop.~3]{MO-2016} and \cite[Prop.~3]{MO-2017}.

When $\alpha<0$, a convenient strategy for the choice of a singular sequence for $H_\alpha$ at some \emph{negative} value $-\lambda$ is described by the following Lemma.

\begin{lemma}\label{lem:strategy-for-essspec}
 Let $m>m^*$ and $\alpha<0$. Assume that a sequence $(\xi_n)_n$ in $\mathcal{D}(\mathcal{A}_{\lambda,\alpha})$, the domain of the operator \eqref{eq:DA_op}, and a number $\lambda>0$ are given such that
 \begin{itemize}
  \item[(a)] there exist constants $c_1,c_2>0$ such that $c_1<\|\xi_n\|_{H^{-\frac{1}{2}}}<c_2$ for all $n$,
  \item[(b)] $\langle\xi_n,W_\lambda\xi_m\rangle_{H^{-\frac{1}{2}},H^{\frac{1}{2}}}\to 0$ as $n,m\to +\infty$ with $n\neq m$,
  \item[(c)] $\|(T_\lambda+\alpha\mathbbm{1})\xi_n\|_{H^{\frac{1}{2}}}\to 0$ as $n\to +\infty$.
 \end{itemize}
 Then the sequence $(g_n)_n$ in $\mathcal{D}(H_\alpha)$ defined by
 \[
  g_n\;:=\;(H_\mathrm{free}+\lambda\mathbbm{1})^{-1} u^\lambda_{\mathcal{A}_{\lambda,\alpha}\xi_n}+u^\lambda_{\xi_n}
 \]
 is a singular sequence for $H_\alpha$ at $-\lambda$, and hence $-\lambda\in\sigma_{\mathrm{ess}}(H_\alpha)$.
\end{lemma}

\begin{proof}
 That $g_n\in\mathcal{D}(H_\alpha)$ is clear from \eqref{eq:Birman_parametr_Halpha}. 
 
 From the assumption (a) and Proposition \ref{prop:space_of_charges}(ii)-(iii),
 \[\tag{i}\label{i-here}
  \begin{split}
   \liminf_n\|u_{\xi_n}^\lambda\|_{\cH}\;&\approx\;\liminf_n\|\xi_n\|_{H^{-\frac{1}{2}}}\;>\;0 \\
   \limsup_n\|u_{\xi_n}^\lambda\|_{\cH}\;&\approx\;\limsup_n\|\xi_n\|_{H^{-\frac{1}{2}}}<\;+\infty\,.
  \end{split}
 \]
 Moreover,
 \[\tag{ii}\label{ii-here}
  \begin{split}
   \big\| u^\lambda_{\mathcal{A}_{\lambda,\alpha}\xi_n}\big\|_{\cH}^2\;&=\;\big\langle \mathcal{A}_{\lambda,\alpha}\xi_n,W_\lambda \mathcal{A}_{\lambda,\alpha}\xi_n \big\rangle_{H^{-\frac{1}{2}},H^{\frac{1}{2}}} \\
   &=\;4\,\big\langle W_\lambda^{-1}(T_\lambda+\alpha\mathbbm{1})\xi_n,(T_\lambda+\alpha\mathbbm{1})\xi_n \big\rangle_{H^{-\frac{1}{2}},H^{\frac{1}{2}}} \\
   &\leqslant\; 4\;\big\|W_\lambda^{-1}\big\|_{H^{\frac{1}{2}}\to H^{-\frac{1}{2}}}\|(T_\lambda+\alpha\mathbbm{1})\xi_n\|_{H^{\frac{1}{2}}}^2\;\xrightarrow[]{\:n\to +\infty\:}\;0
  \end{split}
 \]
having used \eqref{eq:scalar_products} in the first step, \eqref{eq:DA_op} in the second, Proposition \ref{prop:space_of_charges}(i) and the Inverse Mapping Theorem in the third, and assumption (c) in the last. Since, by continuity,
\[\tag{iii}\label{iii-here}
 \big\| (H_\mathrm{free}+\lambda\mathbbm{1})^{-1} u^\lambda_{\mathcal{A}_{\lambda,\alpha}\xi_n}\big\|_{\cH}\;\xrightarrow[]{\:n\to +\infty\:}\;0\,,
\]
then from the first of \eqref{i-here} and from \eqref{iii-here} we deduce that
\[\tag{iv}\label{iv-here}
 \liminf_n\|g_n\|_{\cH}\;>\;0\,.
\]


Next, from the second of \eqref{i-here}, and from the property
\[
 \big\langle u_{\xi_n}^\lambda,u_{\xi_m}^\lambda\big\rangle_\cH\;=\;\langle\xi_n,W_\lambda\xi_m\rangle_{H^{-\frac{1}{2}},H^{\frac{1}{2}}}\;\xrightarrow[n\neq m]{\:n,m\to +\infty\:}\,0
\]
that follows from \eqref{eq:scalar_products} and assumption (b), it is a standard exercise to deduce that
\[
 u_{\xi_n}^{\lambda}\;\rightharpoonup\;0\qquad\textrm{in $\cH$ \quad as $n\to +\infty$}\,,
\]
which, together with \eqref{iii-here}, yields
\[\tag{v}\label{v-here}
 g_n\;\rightharpoonup\;0\qquad\textrm{in $\cH$ \quad as $n\to +\infty$}\,.
\]

Last, from \eqref{eq:actionHalpha2} and \eqref{ii-here},
\[\tag{vi}\label{vi-here}
 \begin{split}
  \|(H_\alpha-(-\lambda)\mathbbm{1}) g_n\|_{\cH}\;&=\;\big\| u^\lambda_{\mathcal{A}_{\lambda,\alpha}\xi_n}\big\|_{\cH}\;\xrightarrow[]{\:n\to +\infty\:}\;0\,.
 \end{split}
\]

Owing to the properties \eqref{iv-here}, \eqref{v-here}, and \eqref{vi-here}, the sequence $(g_n)_n$ is indeed a singular sequence for $H_\alpha$ at $-\lambda$, and hence $-\lambda\in\sigma_{\mathrm{ess}}(H_\alpha)$.
\end{proof}

We can now come to the proof of the main result of this Section.

\begin{proof}[Proof of  Proposition \ref{prop:ess-spec-alphanegative_1}]
 In either case $\mu\in[0,+\infty)$ and $\mu\in[-\frac{\alpha^2}{4\pi^4},0)$ we prove that $\mu\in\sigma_{\mathrm{ess}}(H_\alpha)$  by means of a suitable singular sequence for $H_\alpha$ at $\mu$.

 In the case $\mu\in[0,+\infty)$ we proceed exactly as in the proof of Proposition \ref{prop:ess-spec-alphapositive}, thus considering a sequence $(F_n)_n$ in $H^2_\mathrm{f}(\mathbb{R}^3\times\mathbb{R}^3)$, with $\int_{\mathbb{R}^3}\widehat{F}_n(p,q)\ud p=0$, that is a singular sequence for $H_{\mathrm{free}}$ at $\mu$, and which is therefore also a singular sequence for $H_\alpha$ at $\mu$.
 
 In the case $\mu\in[-\frac{\alpha^2}{4\pi^4},0)$ we set $\lambda:=-\mu$ and we produce a sequence $(\xi_n)_n$ in $H^{\frac{1}{2}}(\mathbb{R}^3)$ with the properties prescribed by Lemma \ref{lem:strategy-for-essspec}, thus deducing that $\mu=-\lambda\in\sigma_{\mathrm{ess}}(H_\alpha)$.

 Let $r_0\geqslant 0$ be the unique root of
 \[
  2\pi^2\sqrt{\nu r_0^2+\lambda}\;=\;|\alpha|
 \]
 for the considered $\lambda\in(0,\frac{\alpha^2}{4\pi^4}]$.
 Correspondingly, let $\xi_n\in H^{1/2}_{\ell=1}(\mathbb{R}^3)$ be defined in polar coordinates $p\equiv|p|\Omega_p$ by
 \[
  \begin{split}
   \widehat{\xi}_n(p)\;&:=\;Y_{1,1}(\Omega_p)f_n(|p|) \\
   f_n(r)\;&:=\;\frac{\,\sqrt{n}\,}{\:r}\,\mathbf{1}_{[r_0+\frac{1}{n},r_0+\frac{2}{n}]}\,.
  \end{split}
 \]
  The charge $\xi_n$ is constructed in view of two crucial properties, the relevance of which will be clear in a moment: the fact that $\xi_n$ has definite angular symmetry $\ell=1$ (any other $\ell\geqslant 1$ would do the job, yet \emph{not} $\ell=0$) and the fact that the sequence $(r^2f_n^2)_n$ is an approximate Dirac distribution in $L^2(\mathbb{R}^+,\ud r)$ centred at $r_0$:
 \[
  r^2f_n(r)^2\;=\;n\,\mathbf{1}_{[r_0+\frac{1}{n},r_0+\frac{2}{n}]}\;\xrightarrow[\:n\to +\infty\:]{\mathcal{D}'(\mathbb{R})}\;\delta(r-r_0)\,.
 \]

 Let us check that $(\xi_n)_n$ satisfies the properties (a)-(b)-(c) of Lemma \ref{lem:strategy-for-essspec}. Since
 \[
  \begin{split}
  \|\xi_n\|_{H^{-\frac{1}{2}}(\mathbb{R}^3)}^2\;&=\;\big\|(r^2+1)^{-\frac{1}{4}}f_n \big\|^2_{L^2(\mathbb{R}^+,r^2\ud r)} \\
  &=\;\int_0^{+\infty}\frac{r^2 f_n(r)^2}{(r^2+1)^{\frac{1}{2}}}\,\ud r\;\xrightarrow[]{\:n\to +\infty\:}\;(r_0^2+1)^{-\frac{1}{2}}\,,
  \end{split}
 \]
 then property (a) of Lemma \ref{lem:strategy-for-essspec} is satisfied. Moreover, analogously to \eqref{eq:Tform_decomp_1}-\eqref{eq:Tform_decomp_2},
 \[
  \begin{split}
   \langle\xi_n,W_\lambda&\xi_m\rangle_{H^{-\frac{1}{2}},H^{\frac{1}{2}}}\;=\;2\pi^2\int_0^{+\infty}\!\!\ud r\,\frac{r^2 f_n(r) f_m(r)}{\sqrt{\nu r^2+\lambda}} \\
   &-2\pi\int_0^{+\infty}\!\!\ud r\int_0^{+\infty}\!\!\ud r'\,r^2f_n(r)\,r'^2f_m(r')\int_{-1}^1\ud y\,\frac{P_{\ell=1}(y)}{\,r^2+r'^2+\mu r r' y+\lambda\,}\,;
  \end{split}
 \]
 therefore, when $n\neq m$, using the property $f_n(r)f_m(r)=0$, one has
 \[
  \begin{split}
   \sqrt{nm} \,\langle\xi_n,W_\lambda\xi_m\rangle_{H^{-\frac{1}{2}},H^{\frac{1}{2}}}\;&=\;-2\pi n m\int_{r_0+\frac{1}{n}}^{r_0+\frac{2}{n}}\!\!\ud r\int_{r_0+\frac{1}{m}}^{r_0+\frac{2}{m}}\!\!\ud r'\int_{-1}^1\ud y\,\frac{r r' y}{\,r^2+r'^2+\mu r r' y+\lambda\,} \\
   &\;\;\;\;\xrightarrow[(n\neq m)]{\:n,m\to +\infty\:}\;-2\pi\int_{-1}^1\ud y\,\frac{r_0^2y}{\,2r_0^2+\mu r_0^2 y+\lambda\,}\,,
  \end{split}
 \]
 whence 
 \[
  \langle\xi_n,W_\lambda\xi_m\rangle_{H^{-\frac{1}{2}},H^{\frac{1}{2}}}\;\xrightarrow[(n\neq m)]{\:n,m\to +\infty\:}\;0\,.
 \]
 Thus, property (b) of Lemma \ref{lem:strategy-for-essspec} is satisfied.

 It remains to check that $\xi_n\in\mathcal{D}(\mathcal{A}_{\lambda,\alpha})$ and property (c)  of Lemma \ref{lem:strategy-for-essspec}. Owing to \eqref{eq:DA_op} and to the fact that $\xi_n\in H^{s}_{\ell=1}(\mathbb{R}^3)$ for any $s\in\mathbb{R}$, to prove that $\xi_n\in\mathcal{D}(\mathcal{A}_{\lambda,\alpha})$ is equivalent to prove that $(T_\lambda+\alpha\mathbbm{1})\xi_n\in H^{\frac{1}{2}}(\mathbb{R}^3)$, and thus it follows from the bound \eqref{eq:T-ell-3/2-1/2_included} of Proposition \ref{prop:regpropTlambda}.
 
 Again in analogy to \eqref{eq:Tform_decomp_1}-\eqref{eq:Tform_decomp_2},
 \[
  \begin{split}
   \|&(T_\lambda+\alpha\mathbbm{1})\xi_n\|_{H^{\frac{1}{2}}(\mathbb{R}^3)}^2\;=\;\int_0^{+\infty}\!\!\ud r\,r^2(r^2+1)^{\frac{1}{2}} \Big|(2\pi^2\sqrt{\nu r^2+\lambda}-|\alpha|)f_n(r) \:+ \\
   &\qquad\qquad \qquad \qquad \qquad  +2\pi\int_0^{+\infty}\!\!\ud r'\,r'^2f_n(r')\int_{-1}^1\ud y\,\frac{P_{\ell=1}(y)}{\,r^2+r'^2+\mu rr'y+\lambda\,}\Big|^2 \\
   &\leqslant\;2\int_0^{+\infty}\!\!\ud r\,(r^2+1)^{\frac{1}{2}} \big|2\pi^2\sqrt{\nu r^2+\lambda}-|\alpha|\,\big|^2\,r^2f_n(r)^2 \\
    &\qquad +2\int_0^{+\infty}\!\!\ud r\,r^2(r^2+1)^{\frac{1}{2}} \Big|\,2\pi\int_0^{+\infty}\!\!\ud r'\,r'^2f_n(r')\int_{-1}^1\ud y\,\frac{y}{\,r^2+r'^2+\mu rr'y+\lambda\,}\Big|^2 \\
    &=:\;2\,(\mathrm{I}_n)+2\,(\mathrm{II}_n)\,.
  \end{split}
 \]
 Since $r^2f_n(r)^2\approx\delta(r-r_0)$ and $2\pi^2\sqrt{\nu r_0^2+\lambda}=|\alpha|$, then $(\mathrm{I}_n)\to 0$ as $n\to +\infty$. On the other hand,
 \[
  \begin{split}
   n\cdot(\mathrm{II}_n)\;&=\;\int_0^{+\infty}\!\!\ud r\,r^2(r^2+1)^{\frac{1}{2}} \Big|\,2\pi n\int_{r_0+\frac{1}{n}}^{r_0+\frac{2}{n}}\!\!\ud r'\,r'\int_{-1}^1\ud y\,\frac{y}{\,r^2+r'^2+\mu rr'y+\lambda\,}\Big|^2
  \end{split}
 \]
 and 
 \[
  \begin{split}
   h_n(r)\;:=\;2\pi n\int_{r_0+\frac{1}{n}}^{r_0+\frac{2}{n}}\!\!\ud r'\,r'&\int_{-1}^1\ud y\,\frac{y}{\,r^2+r'^2+\mu rr'y+\lambda\,} \\
   \xrightarrow[]{\,n\to +\infty\,}\;h_\infty(r)\;:=&\;2\pi r_0\int_{-1}^1\ud y\,\frac{y}{\,r^2+r_0^2+\mu rr_0y+\lambda\,} \\
   =&\;\frac{2\pi}{\mu}\,\frac{1}{r}\,\bigg(2-{\textstyle\frac{r^2+r_0^2+\lambda}{\mu r_0 r}}\:\ln\frac{\;{\textstyle\frac{r^2+r_0^2+\lambda}{\mu r_0 r}}+1\;}{{\textstyle\frac{r^2+r_0^2+\lambda}{\mu r_0 r}}-1} \bigg)\,.
  \end{split}
 \]
 Applying the expansion
 \[
  2-z\ln\frac{z+1}{z-1}\;=\;-{\textstyle\frac{2}{3}}\,(z-1)^{-2}+o\big((z-1)^{-2}\big)\qquad\textrm{as }\,z\to +\infty
 \]
 to $z=\frac{r^2+r_0^2+\lambda}{\mu r_0 r}$ yields 
 \[
  h_\infty(r)\;\sim\;
  \begin{cases}
    \mathrm{const.}\,\times\, r& \textrm{as } r\downarrow 0 \\
    \mathrm{const.}\,\times\,r^{-3}& \textrm{as } r\to +\infty\,,
  \end{cases}
 \]
 whence, by dominated convergence,
 \[
  n\cdot(\mathrm{II}_n)\;\to\;\frac{2\pi}{\mu}\int_0^{+\infty}\!\!\ud r\,r^2(r^2+1)^{\frac{1}{2}}\,|h_\infty(r)|^2\;<\;+\infty\,,
 \]
  and also $(\mathrm{II}_n)\to 0$ as $n\to +\infty$.
 Thus,  $\|(T_\lambda+\alpha\mathbbm{1})\xi_n\|_{H^{\frac{1}{2}}}\to 0$ as $n\to +\infty$, which completes the proof.
\end{proof}

\section{Finiteness of the discrete spectrum}\label{sec:finite_disc_spec}

In this Section we prove Theorem \ref{thm:disc-spec-finite} and we complete the proof of Theorem \ref{thm:ess-spec}.

The proof of Theorem \ref{thm:disc-spec-finite} is based on the following simple idea. We know that the quadratic form of $H_\alpha$ is the sum of the quadratic form of the free Hamiltonian for the regular part and the quadratic form of the charges, more precisely
\[
 (H_\alpha+\lambda\mathbbm{1})[g]\;=\;(H_{\mathrm{free}}+\lambda\mathbbm{1})[F^{\lambda}]+2\langle \xi,(T_\lambda+\alpha\mathbbm{1})\xi\rangle_{H^{\frac{1}{2}},H^{-\frac{1}{2}}}
\]
for every $g=F^{\lambda}+u_\xi^\lambda\in\mathcal{D}[H_\alpha]$. If one proves that for $\lambda=\lambda_0:=\frac{\alpha^2}{4\pi^4}$ the form of the charges is positive except for a finite-dimensional subspace of the charge space, then given the one-to-one correspondence $\xi\leftrightarrow u_\xi^{\lambda_0}$ and the positivity of $H_{\mathrm{free}}$ one deduces that $H_\alpha$ is bounded from below by $-\lambda_0\mathbbm{1}$ but for a residual, finite-dimensional subspace of its form domain in which the bottom may be even lower than $-\lambda_0$; the residual subspace can then only accommodate a finite number of linearly independent eigenfunctions. When subsequently we shall prove that $\sigma_{\mathrm{ess}}(H_\alpha)=[-\lambda_0,+\infty)$, this would allow us to conclude that $\sigma_{\mathrm{disc}}(H_\alpha)$ is finite.

The problem is then boiled down to the analysis of the form of the charges and hence of $T_\lambda$ as a bounded $H^{\frac{1}{2}}(\mathbb{R}^3)\to H^{-\frac{1}{2}}(\mathbb{R}^3)$ map. In fact, because of its angular symmetry, $T_\lambda$ is also a continuous $H_\ell^{1/2}(\mathbb{R}^3)\to H_\ell^{-1/2}(\mathbb{R}^3)$ map separately for each sector $\ell\in\mathbb{N}_0$ of definite angular momentum (Proposition \ref{prop:regpropTlambda}), and since the only possible eigenfunction of $H_\alpha$ for $\alpha<0$ have charges with odd-$\ell$ symmetry (Lemma \ref{lem:only_odd_symm}) it suffices to consider $T_\lambda$ acting on the Hilbert sub-space $H^s_-(\mathbb{R}^3)=\bigoplus_{\mathrm{odd}\:\ell}H_\ell^{-1/2}(\mathbb{R}^3)$. 

We shall make use of a convenient decomposition of $T_\lambda$, in the same spirit of similar decompositions that for a different purpose (the computation of the deficiency indices of $T_\lambda$ as an operator on $L^2(\mathbb{R}^3)$) were introduced by Minlos and Shermatov \cite{Minlos-Shermatov-1989} and Minlos \cite{Minlos-TS-1994,Minlos-2011-preprint_May_2010}.

The precise decomposition we exploit here was introduced recently by Yoshitomi \cite{Yoshitomi_MathSlov2017}, as a clever modification of Minlos's decomposition \cite[Eq.~(2.18)-(2.20)]{Minlos-2011-preprint_May_2010} by means of a suitable localisation in momentum space.
In \cite{Yoshitomi_MathSlov2017} too the goal is to prove the finiteness of the discrete spectrum, but this is achieved only for sufficiently large $m$: in our approach, instead, we are able to control the whole regime $m>m^*$.

For $\lambda>0$ and $n\in\mathbb{N}$ let
$\mathbf{1}_n$ and $\mathbf{1}_n^{\textsc{c}}$ denote, respectively, the characteristic function of the ball $\{p\in\mathbb{R}^3\,|\,|p|\leqslant n\}$ and of its complement, let
\[
 G(p,q)\;:=\;p^2+q^2+\mu\, p\cdot q\,,
\]
and let us introduce the maps $\xi\mapsto R_\lambda\xi$, $\xi\mapsto L_{\lambda,n}\xi$, and $\xi\mapsto Q_{\lambda,n}\xi$ defined on the $\xi$'s in $L^2(\mathbb{R}^3)$ as follows:
\begin{eqnarray}
 \widehat{(R_\lambda\xi)}(p)\!\!\!&:=&\!\!\!\theta_\lambda(p)\,\widehat{\xi}(p) \label{eq:defRlambda}\\
 \widehat{(L_{\lambda,n}\xi)}(p)\!\!\!&:=&\!\!\!\int_{\mathbb{R}^3}\frac{\mathbf{1}_n^{\textsc{c}}(p)\,\mathbf{1}_n^{\textsc{c}}(q)}{\theta_\lambda(p)\,G(p,q)\,\theta_\lambda(q)}\,\widehat{\xi}(q)\,\ud q \\
 \widehat{(Q_{\lambda,n}\xi)}(p)\!\!\!&:=&\!\!\!\int_{\mathbb{R}^3}\mathcal{Q}_{\lambda,n}(p,q)\,\widehat{\xi}(q)\,\ud q\,,
\end{eqnarray}
where
\begin{equation}
\label{eq:rlambda}
 \theta_\lambda(p)\;:=\;\sqrt{\sqrt{\nu p^2+\lambda}-\sqrt{\lambda}\,}
\end{equation}
and
\begin{equation}\label{Q_decomposition}
 \begin{split}
  &\mathcal{Q}_{\lambda,n}(p,q)\;:=\;\lambda\,\frac{\mathbf{1}_n^{\textsc{c}}(p)\,\mathbf{1}_n^{\textsc{c}}(q)}{\,\theta_\lambda(p)\,G(p,q)(G(p,q)+\lambda)\,\theta_\lambda(q)} -\frac{\mathbf{1}_n^{\textsc{c}}(p)\,\mathbf{1}_n(q)}{\,\theta_\lambda(p)\,(G(p,q)+\lambda)\,\theta_\lambda(q)} \!\!\!\! \\
  &\qquad\qquad\qquad -\frac{\mathbf{1}_n(p)\,\mathbf{1}_n^{\textsc{c}}(q)}{\,\theta_\lambda(p)\,(G(p,q)+\lambda)\,\theta_\lambda(q)}-\frac{\mathbf{1}_n(p)\,\mathbf{1}_n(q)}{\,\theta_\lambda(p)\,(G(p,q)+\lambda)\,\theta_\lambda(q)}\,.
 \end{split}
\end{equation}

It is a straightforward check that in terms of the quantities above one has the decomposition
\begin{equation}\label{Tlambda_decomposition}
 T_\lambda\;=\;2\pi^2\sqrt{\lambda}\,\mathbbm{1}+R_\lambda(2\pi^2\mathbbm{1}+L_{\lambda,n}-Q_{\lambda,n})R_\lambda\,.
\end{equation}

Let us collect relevant properties of the maps $R_\lambda$, $L_{\lambda,n}$, and $Q_{\lambda,n}$.

\begin{lemma}\label{lem:Qlambdan}
 For each $\lambda>0$ and $n\in \mathbb{N}$, $Q_{\lambda,n}$ defines a Hilbert-Schmidt and self-adjoint operator on $L^2(\mathbb{R}^3)$.
\end{lemma}

\begin{proof}
We start by proving that $Q_{\lambda,n}$ is Hilbert-Schmidt. Re-writing \eqref{Q_decomposition} self-explanatorily as
\[
 \mathcal{Q}_{\lambda,n}(p,q)\;=\;\lambda\mathcal{Q}^{(1)}_{\lambda,n}(p,q)-\mathcal{Q}^{(2)}_{\lambda,n}(p,q)-\mathcal{Q}^{(3)}_{\lambda,n}(p,q)-\mathcal{Q}^{(4)}_{\lambda,n}(p,q)\,,
\]
it suffices to show that $\mathcal{Q}^{(j)}_{\lambda,n}\in L^2(\mathbb R^3\times \mathbb R^3)$ for $j\in\{1,2,3,4\}$. In the estimates that follow we will systematically use the bounds
\begin{equation}\tag{*}\label{*}
{\textstyle(1-\frac{\mu}{2})}(p^2+q^2)\;\leqslant\; G(p,q) \;\leqslant\; {\textstyle(1+\frac{\mu}{2})} (p^2+q^2)
\end{equation} 
and
\begin{equation}\tag{**}\label{**}
\frac{1}{\theta_{\lambda}^2(q)}\;=\;\frac{\sqrt{\nu q^2+\lambda}+\sqrt{\lambda}}{\nu q^2}\,.
\end{equation}

For $\mathcal Q_{\lambda,n}^{(1)}$ one has
\begin{align*}
\|\mathcal Q^{(1)}_{\lambda,n}&\|_{L^2(\mathbb R^3\times \mathbb R^3)}^2\;=\;\iint_{\substack{|p|>n \\ |q|>n}}\:\frac{\ud p\,\ud q}{\,G(p,q)^2(G(p,q)+\lambda)^2r_{\lambda}^2(p)r_{\lambda}^2(q)\,} \\
&\lesssim\;\int_{n}^{+\infty}\!\!\ud r\int_{n}^{+\infty}\!\!\ud r'\:\frac{(\sqrt{\nu r^2+\lambda}+\sqrt{\lambda})(\sqrt{\nu r'^2+\lambda}+\sqrt{\lambda})}{(r^2+r'^2)^4} \\
&\lesssim\;\int_{n}^{+\infty}\!\!\ud r\int_{n}^{+\infty}\!\!\ud r'\:\frac{r\,r'}{\,(r^2+r'^2)^4} \;\lesssim\; \int_{n}^{+\infty}\ud \rho\,\rho^{-5}\;<\;+\infty\,,
 \end{align*}
 where we used \eqref{*} and \eqref{**} and three-dimensional polar coordinates separately in $p, q$ in the second step, the bound $\sqrt{\nu r^2+\lambda}+\sqrt{\lambda}\lesssim r$ (valid for $r>n$ when $n$ is sufficiently large) in the third step, and finally two-dimensional polar coordinates $(r,s)\equiv(\rho,\theta)$, the positivity of the integrand, and the inclusion
 \[
  \{(r,r')\in\mathbb{R}^2\,|\,r>n,r'>n\}\;\subset\;\{(\rho,\theta)\in [0,+\infty)\times[0,{\textstyle\frac{\pi}{2}}]\subset\mathbb{R}^2\,|\,\rho>n\}
 \]
in the last step. This proves that $\mathcal Q_{\lambda,n}^{(1)}$ belongs to $L^2(\mathbb R^3\times \mathbb R^3)$.

For the kernels $\mathcal Q_{\lambda,n}^{(2)}$, $\mathcal Q_{\lambda,n}^{(3)}$, and $\mathcal Q_{\lambda,n}^{(4)}$ it is convenient to observe that the function $\mathcal{K}:\mathbb R^2\rightarrow \mathbb R^+$ defined by 
\begin{align*}
\mathcal{K}(r,r')\;:=\;\frac{(\sqrt{\nu r^2+\lambda}+\sqrt{\lambda})(\sqrt{\nu r'^2+\lambda}+\sqrt{\lambda})}{(r^2+r'^2+\frac{\lambda}{1-\frac{\mu}{2}})^2}
\end{align*}
is smooth for positive $\lambda$, and by means of \eqref{*}, \eqref{**} and polar coordinates in $p$ and $q$ one has 
\begin{align*}
\|\mathcal Q^{(j)}_{\lambda,n}\|_{L^2(\mathbb R^3\times \mathbb R^3)}^2\;\lesssim\;\iint_{\Omega_j}\ud r\,\ud r'\,\mathcal{K}(r,r')\qquad j\in\{2,3,4\}\,,
\end{align*}
where
\[
 \begin{split}
  \Omega_2\;&:=\;\{(r,r')\in\mathbb{R}^2\,|\,r\geqslant n,r'\leqslant n\} \\
  \Omega_3\;&:=\;\{(r,r')\in\mathbb{R}^2\,|\,r\leqslant n,r'\geqslant n\} \\
  \Omega_4\;&:=\;\{(r,r')\in\mathbb{R}^2\,|\,r\leqslant n,r'\leqslant n\}\,.
 \end{split}
\]
Now, 
\[
  \iint_{\Omega_3}\ud r\,\ud r'\,\mathcal{K}(r,r')\;=\;\iint_{\Omega_2}\ud r\,\ud r'\,\mathcal{K}(r,r')\;\lesssim\;\int_n^{+\infty}\!\!\ud r\,r^{-3}\;<\;+\infty\,,
\]
the first step following from the exchange symmetry $r\leftrightarrow r'$, the second step following from the bound $\mathcal{K}(r,r')\lesssim r^{-3}$ on $\Omega_2$, since $r'$ ranges on the compact set $[0,n]$. This proves that $\mathcal Q_{\lambda,n}^{(2)}$ and $\mathcal Q_{\lambda,n}^{(3)}$ belong to $L^2(\mathbb R^3\times \mathbb R^3)$. Obviously the same conclusion holds for  $\mathcal Q_{\lambda,n}^{(4)}$, since $\Omega_4$ is compact.

Thus, $Q_{\lambda,n}$ is a Hilbert-Schmidt operator. In particular, $Q_{\lambda,n}$ is bounded, and it is also symmetric because its kernel is real-valued and symmetric under the exchange $p\leftrightarrow q$. Then $Q_{\lambda,n}$ is self-adjoint.
\end{proof}

\begin{lemma}\label{lem:Llambdan}
 For each $\lambda>0$ and $n\in \mathbb{N}$, $L_{\lambda,n}$ defines a bounded, positive definite, self-adjoint operator on $L^2(\mathbb{R}^3)$. $L$ is reduced with respect to $L^2(\mathbb{R}^3)=L^2_+(\mathbb{R}^3)\oplus L^2_-(\mathbb{R}^3)$ and eventually
 \begin{equation}\label{eq:norm_of_L}
  \|L_{\lambda,n}\|_{L^2_-(\mathbb{R}^3)\to L^2_-(\mathbb{R}^3)}<\;2\pi^2
 \end{equation}
 for $n$ sufficiently large. In particular, for $n$ sufficiently large the operator $2\pi^2\mathbbm{1}+L_{\lambda,n}$ is invertible on $L^2_-(\mathbb{R}^3)$ with continuous inverse for $n$ sufficiently large.
\end{lemma}

\begin{proof}
 Let us consider the operator $L_{0,n}$ obtained by setting $\lambda=0$ in the definition of $L_{\lambda,n}$, that is, $L_{0,n}=\nu^{-\frac{1}{2}}\mathbf{1}_n^{\textsc{c}}\,L\,\mathbf{1}_n^{\textsc{c}}$, where
 \[
  \widehat{(L\,\xi)}(p)\;:=\;\int_{\mathbb{R}^3}\frac{\widehat{\xi}(q)}{\sqrt{|p|}\,G(p,q)\sqrt{|q|}\,}\,\ud q\,.
 \]

 The map $\xi\mapsto L\,\xi$ is well studied in the literature (see, e.g., \cite[Sec.~4]{Minlos-2012-preprint_30sett2011} or \cite[Sec.~3]{CDFMT-2012}): with respect to the canonical decomposition \eqref{eq:can_decomp}
the operator $L$ is reduced as
 \[
  L\;=\;\bigoplus_{\ell=0}^\infty\;(\mathcal{L}_\ell\otimes\mathbbm{1})\,,
 \]
 namely in components acting non-trivially only on the radial part of $ L^2_{\ell}(\mathbb{R}^3)$. In turn, by means of the Mellin transform
 \[
  \begin{split}
   & \omega:L^2(\mathbb{R}^+,r^2\,\ud r)\stackrel{\cong}{\longrightarrow} L^2(\mathbb{R},\ud \rho) \\
   & (\omega f)(\rho)\,=\,\frac{1}{\sqrt{2\pi}\,}\!\int_0^{+\infty}r^{-\ii\rho+\frac{1}{2}}f(r)\,\ud r\,,
  \end{split}
 \]
 each $\mathcal{L}_\ell$ is unitarily equivalent to $\omega\mathcal{L}_\ell\omega^{-1}$, which is the multiplication operator by the function
 \begin{equation}
 \label{eq:lambda}
  \lambda_\ell(\rho)\;:=\;\begin{cases}
   \displaystyle\;\;\;2\pi^2\!\!\int_0^1 P_\ell(x)\,\frac{\cosh\big(\rho\arcsin(\frac{x}{m+1})\big)}{\,\cosh(\frac{\pi\rho}{2})\cos\big(\arcsin(\frac{x}{m+1})\big)\,}\,\ud x &\quad \textrm{for even $\ell$} \\
   \displaystyle\,-2\pi^2\!\!\int_0^1 P_\ell(x)\,\frac{\sinh\big(\rho\arcsin(\frac{x}{m+1})\big)}{\,\sinh(\frac{\pi\rho}{2})\cos\big(\arcsin(\frac{x}{m+1})\big)\,}\,\ud x & \quad \textrm{for odd $\ell$}\,.
  \end{cases}
 \end{equation}

 Now, $\rho\mapsto\lambda_\ell(\rho)$ is a smooth, even function vanishing at infinity, and moreover when $\ell$ is even (resp., when $\ell$ is odd), it is monotone decreasing (resp., monotone increasing) for $\rho>0$. Thus, $\|\mathcal{L}_\ell\|=|\lambda_\ell(0)|$. A further simple analysis, exploiting the definition of $P_\ell(x)$ and the integration by parts in $x$ (see, e.g., \cite[Lemma 3.5]{CDFMT-2012}), shows that $0\leqslant\lambda_\ell(\rho)\leqslant\lambda_{\ell=0}(\rho)$ for even $\ell$ and 
 $\lambda_{\ell=1}(\rho)\leqslant\lambda_\ell(\rho)\leqslant 0$ for odd $\ell$, and an explicit computation shows that $-\lambda_{\ell=1}(0)\leqslant\lambda_{\ell=0}(0)$. Therefore,
 \[
   \|L\|_{L^2(\mathbb{R}^3)\to L^2(\mathbb{R}^3)}\;=\;\|\mathcal{L}_{\ell=0}\|\;=\;\lambda_{\ell=0}(0)\;=\;2\pi^2(m+1)\arcsin({\textstyle\frac{1}{m+1}})\,.
 \]
 This proves the boundedness of $L$ on $L^2(\mathbb{R}^3)$, as well as its reduction with respect to the sectors of even or of odd angular momenta.
 From the above-mentioned properties of $\lambda_\ell(\rho)$ one also deduces
 \[
 \begin{split}
   \|L\|_{L^2_-(\mathbb{R}^3)\to L^2_-(\mathbb{R}^3)}\;&=\;\|\mathcal{L}_{\ell=1}\|\;=\;|\lambda_{\ell=1}(0)| \\
   &=\;4\pi(m+1)\big(1-\sqrt{m(m+2)}\arcsin({\textstyle\frac{1}{m+1}}) \big) \\
   &=\;2\pi^2\sqrt{\nu}\,\Lambda(m)\,.
 \end{split}
  \]

Next, let us compare the above norm with that of $L_{0,n}$. Since
 \[
  \lim_{|p|\to+\infty}|p|^{-\frac{1}{2}}\theta_\lambda(p)\;=\;\nu^{\frac{1}{4}}\,,
 \]
 then for arbitrary $\varepsilon,\lambda>0$ it is possible to choose $n$ large enough so as
 \[
  \Big\| \mathbf{1}^c_n(p)\mathbf{1}^c_n(q)\Big(\frac{\sqrt{p}}{\theta_\lambda(p)}\,\frac{\sqrt{q}}{\theta_\lambda(q)}-\frac{1}{\sqrt{\nu}}\Big)\Big\|_{L^\infty(\mathbb{R}^3\times\mathbb{R}^3,\ud p\,\ud q)}\leqslant\;\varepsilon\,.
 \]
 Therefore, eventually in $n$,
 \[
  \begin{split}
   \|(L_{\lambda,n}&-L_{0,n})\xi\|_{L^2(\mathbb{R}^3,\ud p)}\;= \\
   &=\;\Big\| \int_{\mathbb{R}^3}\!\ud q\,\frac{\widehat{\xi}(q)}{\sqrt{|p|}\,G(p,q)\sqrt{|q|}\,}\mathbf{1}^c_n(p)\mathbf{1}^c_n(q)\Big(\frac{\sqrt{p}}{\theta_\lambda(p)}\,\frac{\sqrt{q}}{\theta_\lambda(q)}-\frac{1}{\sqrt{\nu}}\Big)\Big\|_{L^2(\mathbb{R}^3,\ud p)} \\
   &\leqslant\;\varepsilon\,\|L\|_{L^2(\mathbb{R}^3)\to L^2(\mathbb{R}^3)}\,\|\xi\|_{L^2(\mathbb{R}^3)}\;\equiv\;\widetilde{\varepsilon}\:\|\xi\|_{L^2(\mathbb{R}^3)}
  \end{split}
 \]
 for arbitrary $\widetilde{\varepsilon}>0$.

 For $\xi\in L^2_-(\mathbb{R}^3)$ we can then conclude
 \[
 \begin{split}
  \|L_{\lambda,n}\xi\|_{L^2_-(\mathbb{R}^3)}\;&\leqslant\; \|(L_{\lambda,n}-L_{0,n})\xi\|_{L^2(\mathbb{R}^3)}+ \|L_{0,n}\xi\|_{L^2_-(\mathbb{R}^3)} \\
  &\leqslant\;\widetilde{\varepsilon}\:\|\xi\|_{L^2_-(\mathbb{R}^3)}+\frac{1}{\sqrt{\nu\,}}\|L\,\xi\|_{L^2_-(\mathbb{R}^3)} \\
  &\leqslant\;(\widetilde{\varepsilon}+2\pi^2\Lambda(m))\|\xi\|_{L^2_-(\mathbb{R}^3)}
 \end{split}
 \]
 eventually for $n$ large: from the arbitrariness of $\widetilde{\varepsilon}$ and the property that $\Lambda(m)<1$ for $m>m^*$, the bound \eqref{eq:norm_of_L} follows.

 Without symmetry restriction on $\xi$, obviously the last estimate above is still true for every $n$, now with a constant that depends on $n$ and is not smaller than $2\pi^2$. This proves the boundedness of $L_{\lambda,n}$ on $L^2(\mathbb{R}^3)$. Its self-adjointness and positive-definiteness is then obvious from the fact that $L_{\lambda,n}$ has a real, positive, symmetric kernel.
 \end{proof}

 Next, let us qualify a convenient Hilbert space in terms of the maps $R_\lambda$ and $L_{\lambda,n}$. To this aim, we observe that the map $(\xi,\eta)\mapsto\langle\xi,\eta\rangle_{H^{1/2}_{R_\lambda}}$, where
   \begin{equation}\label{eq:newscalarproduct}
    \langle\xi,\eta\rangle_{H^{1/2}_{R_\lambda}}\;:=\;\langle R_\lambda\xi,R_\lambda\eta\rangle_{L^2(\mathbb{R}^3)}\,,
   \end{equation}
  defines a positive definite inner product in $H^{\frac{1}{2}}(\mathbb{R}^3)$: indeed,
  \[
   \langle\xi,\eta\rangle_{H^{1/2}_{R_\lambda}}\;=\;\int_{\mathbb{R}^3}\overline{\widehat{\xi}(p)}\,\widehat{\eta}(p)\,\theta_\lambda(p)^2\ud p
  \]
 and $\theta_\lambda(p)>0$ for all $p\neq 0$ and $\theta_\lambda(p)\sim\langle p\rangle^{\frac{1}{4}}$ for large $|p|$'s.

 As a consequence, the space $(H^{1/2}_{R_\lambda}\!(\mathbb{R}^3),\langle\cdot,\cdot\rangle_{H^{1/2}_{R_\lambda}}\!)$ defined as the \emph{completion} of $H^{\frac{1}{2}}(\mathbb{R}^3)$ with respect to \eqref{eq:newscalarproduct} is a Hilbert space, and the Bounded Linear Transformation theorem ensures that the $H^{1/2}(\mathbb{R}^3)\to L^2(\mathbb{R}^3)$ isometry $\xi\mapsto R_\lambda\xi$ lifts to an isomorphism $H^{1/2}_{R_\lambda}\!(\mathbb{R}^3)\to L^2(\mathbb{R}^3)$ between Hilbert spaces, which we shall keep denoting by $R_\lambda$. Moreover, $R_\lambda$ acting on $H^{\frac{1}{2}}(\mathbb{R}^3)$ as the Fourier multiplier by a radial function, each component $H^{1/2}_{R_\lambda,\ell}(\mathbb{R}^3)$ of definite angular symmetry is naturally defined, and with obvious meaning of symbols one has
   \begin{equation}\label{eq:weighted_ang_decomp}
     H^{1/2}_{R_\lambda}\!(\mathbb{R}^3)\;\cong\;\bigoplus_{\ell=0}^\infty H^{1/2}_{R_\lambda,\ell}(\mathbb{R}^3)\;\cong\; H^{1/2}_{R_\lambda,+}(\mathbb{R}^3)\oplus  H^{1/2}_{R_\lambda,-}(\mathbb{R}^3)\,.
   \end{equation}

 We have the following.
 
 \begin{lemma}\label{lem:Hlambdan}
  Let $\lambda>0$.
  \begin{itemize}
   \item[(i)] For sufficiently large $n\in\mathbb{N}$, 
    \begin{equation}\label{eq:ln_scalar_product}
  \langle \xi,\eta\rangle_{\cH_{\lambda,n}}\;:=\;\langle R_\lambda \xi,(2\pi^2\mathbbm{1}+L_{\lambda,n})R_\lambda \eta \rangle_{L^2(\mathbb{R}^3)}
 \end{equation}
 defines an equivalent scalar product in $H^{1/2}_{R_\lambda,-}(\mathbb{R}^3)$.
 \item[(ii)] For sufficiently large $n\in\mathbb{N}$, let $(\cH_{\lambda,n},\langle\cdot,\cdot\rangle_{\cH_{\lambda,n}})$ be the Hilbert space coinciding with $H^{1/2}_{R_\lambda,-}(\mathbb{R}^3)$ as a vector space and equipped with the scalar product \eqref{eq:ln_scalar_product}. Then, the map
 \begin{equation}\label{eq:sp-isom}
  \begin{split}
   \cH_{\lambda,n}\;&\stackrel{\cong}{\longrightarrow}\;H^{1/2}_-(\mathbb{R}^3) \\
   \xi\;&\longmapsto\; (2\pi^2+L_{\lambda,n})^{\frac{1}{2}}R_\lambda\xi
  \end{split}
 \end{equation}
 is a Hilbert space isomorphism.
  \end{itemize}
 \end{lemma}

\begin{proof}
As argued already, $R_\lambda:H^{1/2}_{R_\lambda}\!(\mathbb{R}^3)\to L^2(\mathbb{R}^3)$ is an isomorphism between Hilbert spaces.  Let $n$ be large enough so as to make the positive definite operator $2\pi^2\mathbbm{1}+L_{\lambda,n}$ invertible on $L^2_-(\mathbb{R}^3)$ with bounded inverse (Lemma \ref{lem:Llambdan}). In this case \eqref{eq:ln_scalar_product} defines a positive definite inner product in $H^{1/2}_{R_\lambda}\!(\mathbb{R}^3)$. Moreover,
 \[
  \begin{split}
   \|\xi\|_{\cH_{\lambda,n}}^2\;&=\;\langle\xi,\xi\rangle_{\cH_{\lambda,n}}\;=\;\|(2\pi^2+L_{\lambda,n})^{\frac{1}{2}}R_\lambda\xi\|_{L^2}^2 \\
   &\approx\;\|R_\lambda\xi\|_{L^2}^2 \;=\;\|\xi\|^2_{H^{1/2}_{R_\lambda}}\,,
  \end{split}
 \]
which shows that the scalar products $\langle\cdot,\cdot\rangle_{\cH_{\lambda,n}}$ and $\langle\cdot,\cdot\rangle_{H^{1/2}_{R_\lambda}}$ are equivalent. This proves part (i). Part (ii) is an obvious consequence of (i) and of formula \eqref{eq:ln_scalar_product} in particular.
\end{proof}

We can now prove Theorem \ref{thm:disc-spec-finite}.

\begin{proof}[Proof of Theorem \ref{thm:disc-spec-finite}]
We know from Lemmas \ref{lem:reduction_lemma} and \ref{lem:only_odd_symm} that if $H_\alpha g=-Eg$ for some non-zero $g\in\mathcal{D}(H_\alpha)$ and some $-E<-\frac{\alpha^2}{4\pi^4}$, then necessarily $g=u_\xi^E$ for a suitable charge $\xi\in H^{1/2}_-(\mathbb{R}^3)$, the sector with odd angular momenta. For generic $\lambda>0$ one can decompose $g=u_\xi^E=F^\lambda+u_\xi^\lambda$, the charge $\xi$ being the same, with the same angular symmetry. This will justify why it will suffice to consider, in due time, only the elements $F^\lambda+u_\xi^\lambda\in\mathcal{D}[H_\alpha]$ with $\xi \in H^{1/2}_-(\mathbb{R}^3)$. Let us denote such subspace with $\mathcal{D}_\mathrm{odd}[H_\alpha]$.

Let now $n\in\mathbb{N}$ be large enough as prescribed by Lemmas \ref{lem:Llambdan} and \ref{lem:Hlambdan}.

 $Q_{\lambda,n}$ is a compact operator on $L^2_-(\mathbb{R}^3)$ (Lemma \ref{lem:Qlambdan}), and so is
 \[
  (2\pi^2+L_{\lambda,n})^{-\frac{1}{2}}Q_{\lambda,n}(2\pi^2+L_{\lambda,n})^{-\frac{1}{2}}
 \]
because $(2\pi^2+L_{\lambda,n})^{-\frac{1}{2}}$ is $L^2_-(\mathbb{R}^3)$-bounded (Lemma \ref{lem:Llambdan}). The Hilbert space isomorphism \eqref{eq:sp-isom} then preserves the compactness property and the operator
 \[
  Z_{\lambda,n}\;:=\;R_\lambda^{-1}(2\pi^2+L_{\lambda,n})^{-1}Q_{\lambda,n}\,R_{\lambda}
 \]
 is compact on $\cH_{\lambda,n}$. $Z_{\lambda,n}$ is also self-adjoint, because
 \[
  \langle \xi,Z_{\lambda,n} \eta\rangle_{\cH_{\lambda,n}}\;=\;\langle R_\lambda \xi,Q_{\lambda,n}R_\lambda \eta\rangle_{L^2}
 \]
 and  $Q_{\lambda,n}$ is self-adjoint on $L^2_-(\mathbb{R}^3)$.

 From the latter identity one also deduces
 \[
   \langle \xi,(\mathbbm{1}-Z_{\lambda,n})\, \eta\rangle_{\cH_{\lambda,n}}\;=\;\langle R_\lambda \xi,(2\pi^2\mathbbm{1}+L_{\lambda,n}-Q_{\lambda,n})R_\lambda \eta\rangle_{L^2}\,.
 \]
 Specialising the identity above for $\lambda=\frac{\alpha^2}{4\pi^4}$, which will be assumed for the rest of this proof, and combining it with the decomposition \eqref{Tlambda_decomposition}, yields
 \[
  \langle \xi,(T_\lambda+\alpha\mathbbm{1})\xi\rangle_{H^{\frac{1}{2}},H^{-\frac{1}{2}}}\;=\; \langle \xi,(\mathbbm{1}-Z_{\lambda,n})\, \xi\rangle_{\cH_{\lambda,n}}
 \]
 (recall that $\alpha<0$, thus $2\pi^2\sqrt{\lambda}=-\alpha$).

 Now, since $Z_{\lambda,n}$ is compact and self-adjoint, the space
 \[
  X_n\;:=\;\big\{\xi\in \cH_{\lambda,n}\,|\,\langle \xi,(\mathbbm{1}-Z_{\lambda,n})\, \xi\rangle_{\cH_{\lambda,n}}\leqslant 0  \big\}
 \]
 has finite dimension. It is then convenient to decompose
 \[
  \cH_{\lambda,n}\;=\;X_n^\perp\; \oplus_{\cH_{\lambda,n}}\: X_n
 \]
 and, correspondingly, from \eqref{eq:DHab_form},
 \[
  \begin{split}
   \mathcal{D}_\mathrm{odd}[H_\alpha]\;=&\;\mathcal{D}_n^+\dotplus\mathcal{D}_n^- \\
   \mathcal{D}_n^+\;:=&\;H^1_\mathrm{f}(\mathbb{R}^3) \dotplus \{ u_\xi^\lambda\,|\,\xi\in X_n^\perp\cap H^{\frac{1}{2}}(\mathbb{R}^3)\} \\
   \mathcal{D}_n^-\;:=&\;\{ u_\xi^\lambda\,|\,\xi\in X_n\cap H^{\frac{1}{2}}(\mathbb{R}^3)\}\,.
  \end{split}
 \]
 By construction, for every $F^\lambda+u_\xi^\lambda\in\mathcal{D}_n^+$ one has
 \[
  \begin{split}
    (H_\alpha+\lambda\mathbbm{1})[F^\lambda+u_\xi^\lambda]\;&=\;(H_{\mathrm{free}}+\lambda\mathbbm{1})[F^\lambda]+\langle \xi,(T_\lambda+\alpha\mathbbm{1})\xi\rangle_{H^{\frac{1}{2}},H^{-\frac{1}{2}}} \\
    &=\;(H_{\mathrm{free}}+\lambda\mathbbm{1})[F^\lambda]+\langle \xi,(\mathbbm{1}-Z_{\lambda,n})\, \xi\rangle_{\cH_{\lambda,n}}\;\geqslant\;0
  \end{split}
 \]
and $\dim \mathcal{D}_n^-<+\infty$. Thus, $H_\alpha$ is bounded from below by $-\frac{\alpha^2}{4\pi^4}$ (i.e., $-\lambda$) on $\mathcal{D}^+_n$, whereas if $\mathcal{D}_n^-\neq\{0\}$, then on such \emph{finite-dimensional} space the form of $H_\alpha$ attains values that are necessarily below the threshold $-\frac{\alpha^2}{4\pi^4}$.

This implies that $E_{H_\alpha}((-\infty,-\frac{\alpha^2}{4\pi^4}))=E_{H_\alpha}([-\frac{\alpha^2}{4\pi^4(1-\Lambda(m))^2},-\frac{\alpha^2}{4\pi^4}))$ is a \emph{finite rank} orthogonal projection, where $\ud  E_{H_\alpha}$ is the spectral measure associated with $H_\alpha$, and then  also $\sigma(H_\alpha)\cap [-\frac{\alpha^2}{4\pi^4(1-\Lambda(m))^2},-\frac{\alpha^2}{4\pi^4})\subset \sigma_{\mathrm{disc}}(H_\alpha)$.

In order to complete the proof, we need to use the latter information so as to complete the identification of $\sigma_{\mathrm{ess}}(H_\alpha)$, which we are going to do now.
\end{proof}


\begin{proof}[Proof of Theorems \ref{thm:ess-spec} and \ref{thm:disc-spec-finite} -- conclusion]
 When $\alpha\geqslant 0$, Theorem \ref{thm:ess-spec} is entirely proved by Proposition \ref{prop:ess-spec-alphapositive}. When $\alpha<0$, we know that 
 \begin{itemize}
  \item  $\sigma(H_\alpha)\subset[-\frac{\alpha^2}{4\pi^4(1-\Lambda(m))^2},+\infty)$ (from Theorem \ref{thm:properties_of_Halpha}),
 \item $\sigma_{\mathrm{ess}}(H_\alpha)\supset[-\frac{\alpha^2}{4\pi^4},+\infty)$ (from Proposition \ref{prop:ess-spec-alphanegative_1}),
 \item $\sigma_{\mathrm{ess}}(H_\alpha)\cap [-\frac{\alpha^2}{4\pi^4(1-\Lambda(m))^2},-\frac{\alpha^2}{4\pi^4})=\emptyset$ (from the first part of the proof of Theorem \ref{thm:disc-spec-finite}).
 \end{itemize}
 Then necessarily
 \[
  \begin{split}
   \sigma_{\mathrm{ess}}(H_\alpha)\;&=\;[{\textstyle-\frac{\alpha^2}{4\pi^4}},+\infty)\\
   \sigma_{\mathrm{disc}}(H_\alpha)\;&\subset\; [{\textstyle-\frac{\alpha^2}{4\pi^4(1-\Lambda(m))^2}},{\textstyle-\frac{\alpha^2}{4\pi^4}}) \,,
  \end{split}
 \]
 which proves Theorem \ref{thm:ess-spec} also for $\alpha<0$. From the first part of the proof of Theorem \ref{thm:disc-spec-finite} we then conclude that the whole $\sigma_{\mathrm{disc}}(H_\alpha)$ is finite, and this completes the proof of  Theorem \ref{thm:disc-spec-finite}.
\end{proof}

\section{Regimes of absence of bound states}\label{sec:noEV}

In this Section and in the following one we develop the proof of Theorem \ref{thm:EVexist}. 
In this Section, in particular, we discuss the regimes of \emph{absence} of bound states in the discrete spectrum and the two main results are going to be the following.

\begin{proposition}\label{prop:noEV}
Let $\alpha<0$. There exists a mass $M_\star>m^*$, with $M_\star\leqslant(2.617)^{-1}$, such that when $m>M_\star$ there are no bound states of $H_\alpha$ below $\inf\sigma_{\mathrm{ess}}(H_\alpha)$ whose eigenfunctions have the charge $\xi$ with angular symmetry $\ell=1$.
\end{proposition}

\begin{proposition}\label{prop:noEVodd3}
Let $\alpha<0$. The Hamiltonian $H_\alpha$ has \emph{no} bound states below $\inf\sigma_{\mathrm{ess}}(H_\alpha)$ whose eigenfunctions have the charge $\xi$ with angular symmetry $\ell=3,5,7,\dots$
\end{proposition}

Let us present the strategy for proving Propositions \ref{prop:noEV} and \ref{prop:noEVodd3}.

 Owing to the angular decompositions of Section \ref{sec:angular_decomp} and the reduction Lemmas \ref{lem:reduction_lemma} and \ref{lem:scaling}, in order to prove that for chosen $\alpha<0$, $m>m^*$, and $\ell\in\mathbb{N}_0$ the operator $H_\alpha$ does not admit eigenvalues below $\inf\sigma_{\mathrm{ess}}(H_\alpha)$ whose eigenfunctions have the charge $\xi$ with angular symmetry $\ell$, it suffices to show that  $T_1\geqslant 2\pi^2\mathbbm{1}$ as an inequality between symmetric operators on the Hilbert subspace $L^2_\ell(\mathbb{R}^3)$.

 In a sense, it is precisely when investigating the positivity of $T_1-2\pi^2\mathbbm{1}$ that the Minlos-Yoshitomi decomposition \eqref{Tlambda_decomposition} emerges most naturally, since
  \[
  \begin{split}
   \widehat{({T}_1\xi)}(p)-2\pi^2\widehat{\xi}(p)\;&=\;2\pi^2(\sqrt{\nu p^2+1}-1)\widehat{\xi}(p)+\int_{\mathbb{R}^3}\frac{\widehat{\xi}(p)}{p^2+q^2+\mu p\cdot q+1}\,\ud q \\
   &=\;\mathcal{F}R_1(2\pi^2\mathbbm{1}+S_1)R_1\xi\,,
  \end{split}
 \]
 where $\widehat{(R_1\xi)}(p)=\theta_1(|p|)\widehat{\xi}(p)$ and $\theta_1(r)=\sqrt{\sqrt{\nu r^2+1}-1}$ as in \eqref{eq:defRlambda}/\eqref{eq:rlambda}, and 
 \begin{equation*}
  \widehat{(S_1\xi)}(p)\;:=\;\int_{\mathbb{R}^3}\frac{\widehat{\xi}(q)}{\,\theta_1(|p|)(p^2+q^2+\mu p\cdot q+1)\theta_1(|q|)\,}\,\ud q\,,
 \end{equation*}
 therefore $T_1-2\pi^2\mathbbm{1}\geqslant\mathbbm{O}$ on $L^2_\ell(\mathbb{R}^3)$ if
 \begin{equation*}
  \|S_1\|_{L^2_\ell(\mathbb{R}^3)\to L^2_\ell(\mathbb{R}^3)}\;<\;2\pi^2\,.
 \end{equation*}

 As follows from the angular decomposition \eqref{eq:xi-decomp}/\eqref{eq:addition_formula}, the latter condition is equivalent to
 \begin{equation}\label{eq:S1normbelow1}
  \|\mathcal{K}_\ell^{(m)}\|_{L^2(\mathbb{R}^+,r^2\ud r)\to L^2(\mathbb{R}^+,r^2\ud r)}\;<\;1\,,
 \end{equation}
 where $\mathcal{K}_\ell^{(m)}$ is the integral operator with kernel
 \begin{equation}
  \begin{split}
  \label{eq:kernel}
   \mathcal{K}_\ell^{(m)}(r,r')\;:=&\;\frac{1}{\,\pi\,\theta_1(r) \theta_1(r')}\int_{-1}^1\!\ud y\,\frac{P_\ell(y)}{\,r^2+r'^2+\mu r r' y+1\,} \\
   =&\;\frac{1}{\,\pi\mu\,r r' \theta_1(r) \theta_1(r')}\,\phi_\ell\big({\textstyle\frac{r^2+r'^2+1}{\mu r r'}}\big)\,,
  \end{split}
 \end{equation}
where we set for convenience
 \begin{equation}
  \phi_\ell(z)\,:=\;\int_{-1}^1\!\ud y\,\frac{P_\ell(y)}{y+z}\,,\qquad z>1\,.
 \end{equation}

  Since $\frac{r^2+r'^2+1}{\mu r r'}\geqslant\frac{2}{\mu}>m^*+1$, $\phi_\ell(z)$ is indeed well defined, and moreover it satisfies the following properties.
 
 \bigskip
 
 \begin{lemma}\label{lem:phiell}~
  \begin{itemize}
    \item[(i)] For every $z>1$,
  \begin{equation}\label{eq:phiellIneq}
   \phi_1(z)\;<\;\phi_3(z)\;<\;\phi_5(z)\;\cdots\;<\;0\,.
  \end{equation}
   \item[(ii)] For odd $\ell$'s the function $(1,+\infty)\ni z\mapsto\phi_\ell(z)$ is smooth and strictly monotone increasing, with 
   \begin{equation}\label{eq:phiellexpansion}
    \phi_\ell(z)\;=\;-\frac{\,2^{-\ell}\int_{-1}^1\ud y\,(1-y^2)^\ell}{(z-1)^{\ell+1}}\,(1+o(1))\qquad\textrm{as }z\to +\infty\,.
   \end{equation}
   \item[(iii)] For odd $\ell$'s one has the bounds
   \begin{eqnarray}
    |\phi_\ell(z)|\!\!&\leqslant&\!\!\frac{\,2^{-\ell}\int_{-1}^1\ud y\,(1-y^2)^\ell}{(z-1)^{\ell+1}}\qquad\quad\; z>1 \label{eq:boundonphiell} \\
    |\phi_\ell(z)|\!\!&\leqslant&\!\!\frac{\,C_\ell\,2^{-\ell}\int_{-1}^1\ud y\,(1-y^2)^\ell}{z^{\ell+1}}\qquad z>1+m^*\,, \label{eq:boundonphiell-corrected}
   \end{eqnarray}
   where 
   \begin{equation}\label{eq:correcting_factor}
    C_\ell\;:=\;\frac{\,2^\ell\,(1+m^*)^{\ell+1}\int_{-1}^1\ud y\,\frac{-P_\ell(y)}{\,y+1+m^*\,}}{\int_{-1}^1\ud y\,(1-y^2)^\ell}\,.
   \end{equation}
   In particular,
   \[
    C_1\;\approx\;2.74\,,\qquad C_3\;\approx\;6.07\,,\qquad\textrm{etc.}
   \]
\end{itemize}
 \end{lemma}

 \begin{proof}
  Let $\ell\in2\mathbb{N}_0+1$ and $z>1$. Then, following from \eqref{eq:LegendrePoly} and integration by parts,
  \[
    \phi_\ell(z)\;=\;\int_{-1}^1\!\ud y\,\frac{1}{y+z}\Big(\frac{1}{2^\ell \ell!}\,\frac{\ud^\ell}{\ud y^\ell}(y^2-1)^\ell\Big)\;=\;\frac{1}{2^\ell}\int_{-1}^{1}\!\ud y\,\frac{\,(y^2-1)^\ell}{\,(y+z)^{\ell+1}}\;<\;0\,,
  \]
 which proves that all quantities in \eqref{eq:phiellIneq} are strictly negative. Moreover,
 \[
  \frac{\,(y^2-1)^{\ell+2}}{\,2^{\ell+2}(y+z)^{\ell+3}}\;>\;\frac{\,(y^2-1)^\ell}{\,2^\ell(y+z)^{\ell+1}}\,,
 \]
because this is equivalent to $(y^2-1)^2<4(y+z)^2$ and hence also to $(y+1)^2+2(z-1)>0$: this shows that $\phi_{\ell+2}(z)>\phi_\ell(z)$ and completes the proof of \eqref{eq:phiellIneq}. 
From the representation 
\[
  \phi_\ell(z)\;=\;-\frac{1}{2^\ell}\int_{-1}^{1}\!\ud y\,\frac{\,(1-y^2)^\ell}{\,(y+z)^{\ell+1}}
\]
the smoothness and the strictly increasing monotonicity of $z\mapsto\phi_\ell(z)$ are obvious, and one also deduces
\[
 \begin{split}
  \phi_\ell(z)\;&=\;-\frac{1}{\,2^\ell(z-1)^{\ell+1}}\int_{-1}^1\!\ud y\,\frac{(1-y^2)^{\ell}}{(1+\frac{y+1}{z-1})^{\ell+1}}\;\geqslant\;-\frac{1}{\,2^\ell(z-1)^{\ell+1}}\int_{-1}^1\!\ud y\,(1-y^2)^{\ell}
 \end{split}
\]
whence \eqref{eq:phiellexpansion} and \eqref{eq:boundonphiell}. Last, we look for a multiple $C_\ell$ of the function
\[
 (1,+\infty)\,\ni\,z\;\mapsto\;\frac{\,2^{-\ell}\int_{-1}^1\ud y\,(1-y^2)^\ell}{z^{\ell+1}}\;=:\;g_\ell(z)
\]
such that $C_\ell \,g_\ell(z)\geqslant|\phi_\ell(z)|$ in the \emph{restricted regime} $z\in[1+m^*,+\infty)$. 
Obviously the optimal $C_\ell$ is determined by the condition $C_\ell \,g_\ell(1+m^*)=|\phi_\ell(1+m^*)|$, whence \eqref{eq:boundonphiell-corrected} and \eqref{eq:correcting_factor}.
 \end{proof}
 

 \begin{corollary}\label{cor:Kells}
  One has
  \begin{equation}
   \|\mathcal{K}_1^{(m)}\|\;>\;\|\mathcal{K}_3^{(m)}\|\;>\;\|\mathcal{K}_5^{(m)}\|\;>\;\cdots
  \end{equation}
 \end{corollary}

 We shall study the occurrence of condition \eqref{eq:S1normbelow1} by means of Schur's test.

%
%

 
 \begin{proof}[Proof of Proposition \ref{prop:noEV}]
First we show that there is a $m$-dependent constant $C(m)$ such that 
  \[\tag{*}\label{eq:theestimate1}
    \int_0^{+\infty} \!\!\ud r'\,r'^2\,|\mathcal{K}_1^{(m)}(r,r')|\,\frac{\theta_1(r')}{r'^2}\;\leqslant\;C(m)\,\frac{\theta_1(r)}{r^2}\,,
  \]
 which implies, by Schur's test and owing to the symmetry of the kernel $\mathcal{K}_1^{(m)}(r,r')$, that $\|\mathcal{K}_1^{(m)}\|<C(m)$. Since in the present context $z=\frac{r^2+r'^2+1}{\mu r r'}>1+m^*$, we can apply formulas \eqref{eq:boundonphiell-corrected}-\eqref{eq:correcting_factor} that yield, for $\ell=1$,
 \[
  |\phi_1(z)|\;\leqslant\;C_1\cdot\frac{2}{3\,z^2}\quad\forall z>1+m^*\,,\qquad C_1\,\approx\,2.74\,.
 \]
Therefore,
  \[
   \begin{split}
    \int_0^{+\infty} \!\!\ud r'\,r'^2\,|\mathcal{K}_1^{(m)}(r,r')|\,\frac{\theta_1(r')}{r'^2}\;&=\frac{1}{\,\pi\mu}\int_0^{+\infty}\!\!\ud r'\,\frac{1}{\,rr' \theta_1(r)}\,\big|\phi_1\big({\textstyle\frac{r^2+r'^2+1}{\mu r r'}}\big)\big| \\
    &\leqslant\;\frac{2\mu\, C_1 }{3\pi}\,\frac{r}{\theta_1(r)}\int_0^{+\infty}\!\!\ud r'\,\frac{r'}{\,(r^2+r'^2+1)^2} \\
    &=\;\frac{\,\mu C_1}{3 \pi }\,\frac{r^3}{\,\theta_1(r)^2 (1+r^2)}\,\frac{\theta_1(r)}{r^2}\,.
   \end{split}
  \]
  The function $(0,+\infty)\ni r\mapsto A(r):=\frac{r^3}{\,\theta_1(r)^2(1+r^2)}$ attains its maximum at $r=r_\mathrm{max}:=\sqrt{-1+2\nu + 2 \sqrt{1-\nu+\nu^2}}$ and 
 \[
  A(r_\mathrm{max})\;=\;\textstyle\frac{\left(2 \sqrt{\nu ^2-\nu +1}+2 \nu -1\right)^{3/2}}{\left(2 \sqrt{\nu ^2-\nu
   +1}+2 \nu \right) \left(\sqrt{\nu  \left(2 \sqrt{\nu ^2-\nu +1}+2 \nu
   -1\right)+1}-1\right)}\,.
 \]
 Estimate \eqref{eq:theestimate1} is thus proved with
   \[
  C(m)\;=\;\frac{\,\mu C_1}{3 \pi }\,A(r_\mathrm{max})\,,
 \]
 the dependence on $m$ being present both in $\mu$ and in $A(r_\mathrm{max})$, which are both continuously and strictly \emph{decreasing} in $m$. We calculate (see Fig.~\ref{fig:noEV_ell1})
 \[
  C(m)\;=\;1\qquad\Leftrightarrow\qquad m\,\approx\,(2.617)^{-1}\,,
 \]
which shows that there exists a threshold $M_\star\leqslant(2.617)^{-1}$ such that when $m>M_\star$ one has $\|\mathcal{K}_1^{(m)}\|<1$ and then the sufficient condition \eqref{eq:S1normbelow1} for the absence of eigenvalues is satisfied.
\end{proof}

\begin{figure}
\includegraphics[width=7cm]{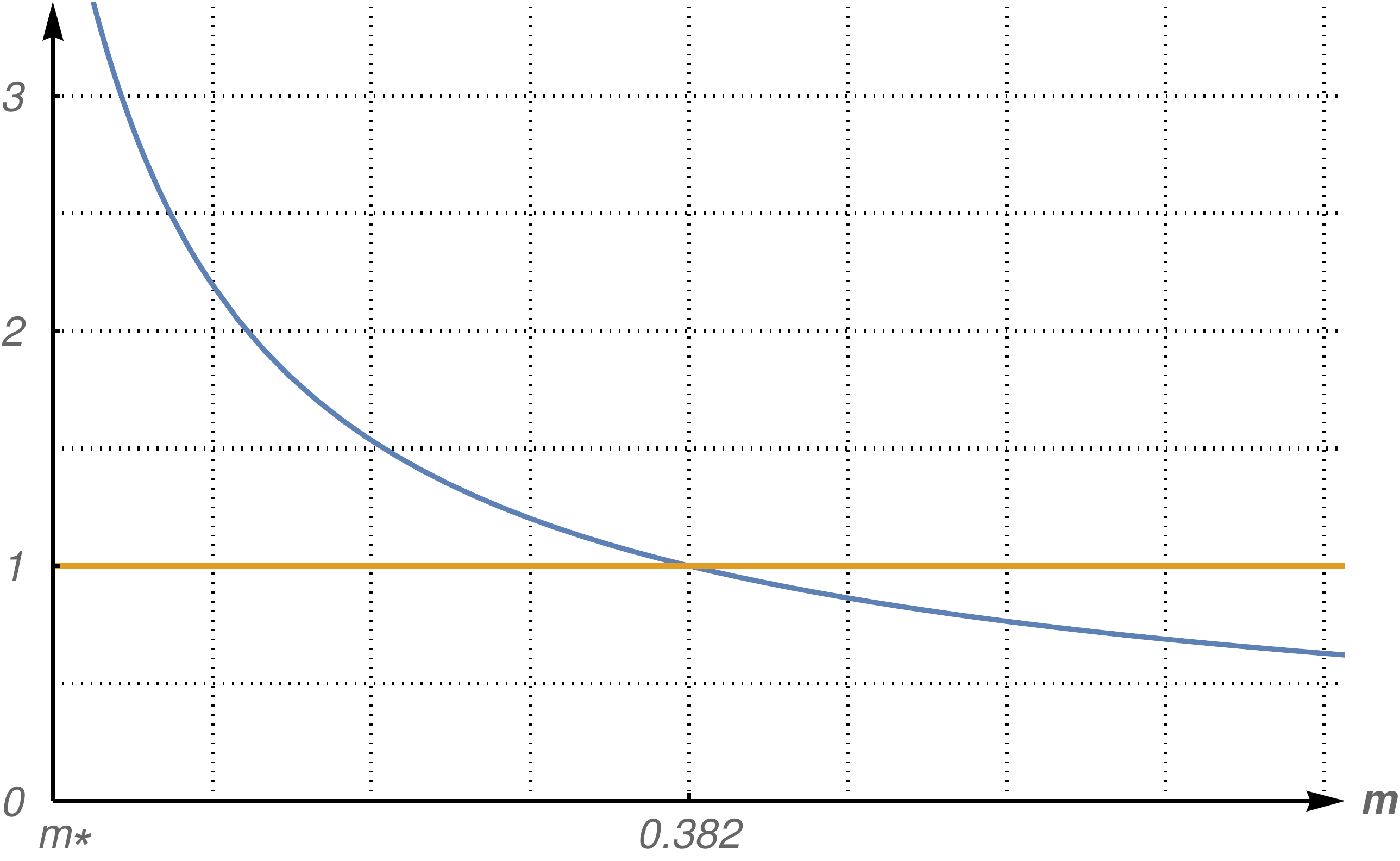}
\caption{Numerical solution to the equation $C(m)=1$ set up in the proof of Proposition \ref{prop:noEV}: the intersection of the function $C(m)$ (blue curve) with the reference value (orange line) takes place at $m\approx 0.382=(2.617)^{-1}.$}\label{fig:noEV_ell1}
\end{figure}

 \begin{proof}[Proof of Proposition \ref{prop:noEVodd3}]
  It is enough to prove that $\|\mathcal{K}_3^{(m)}\|<1$ for all $m\geqslant m^*$ and then to apply Corollary \ref{cor:Kells}. To this aim, let us first show that there is a $m$-dependent constant $C(m)$ such that 
  \[\tag{*}\label{eq:theestimate}
    \int_0^{+\infty} \!\!\ud r'\,r'^2\,|\mathcal{K}_3^{(m)}(r,r')|\,\frac{\theta_1(r')}{r'}\;\leqslant\;C(m)\,\frac{\theta_1(r)}{r}\,,
  \]
 which implies, by Schur's test and owing to the symmetry of the kernel $\mathcal{K}_3^{(m)}(r,r')$, that $\|\mathcal{K}_3^{(m)}\|<C(m)$.  When $\ell=3$  formulas \eqref{eq:boundonphiell-corrected}-\eqref{eq:correcting_factor} yield
 \[
  |\phi_3(z)|\;\leqslant\;C_3\cdot\frac{4}{35\,z^4}\quad\forall z>1+m^*\,,\qquad C_3\,\approx\,6.07\,,
 \]
whence
  \[
   \begin{split}
    \int_0^{+\infty} \!\!\ud r'\,r'^2\,|\mathcal{K}_3^{(m)}(r,r')|\,\frac{\theta_1(r')}{r'}\;&=\frac{1}{\,\pi\mu}\int_0^{+\infty}\!\!\ud r'\,\frac{1}{\,r \theta_1(r)}\,\big|\phi_3\big({\textstyle\frac{r^2+r'^2+1}{\mu r r'}}\big)\big| \\
    &\leqslant\;\frac{\,4\mu^3C_3}{35\pi}\,\frac{r^3}{\theta_1(r)}\int_0^{+\infty}\!\!\ud r'\,\frac{r'^4}{\,(r^2+r'^2+1)^4} \\
    &=\;\frac{\,\mu^3C_3}{\,280\,}\,\frac{r^4}{\,\theta_1(r)^2(r^2+1)^{\frac{3}{2}}}\,\frac{\theta_1(r)}{r}\,.
   \end{split}
  \]
 The function $(0,+\infty)\ni r\mapsto A(r):=\frac{r^4}{\,\theta_1(r)^2(r^2+1)^{3/2}}$ attains its maximum at $r=r_\mathrm{max}:=\sqrt{\frac{3}{2}\sqrt{9\nu^2-4\nu+4\,}+\frac{9}{2}\nu-1}$ and 
 \[
  A(r_\mathrm{max})\;=\;\textstyle\frac{\sqrt{\frac{2}{27}} \left(3 \sqrt{9 \nu ^2-4 \nu +4}+9 \nu -2\right)^2}{
   \left(\sqrt{9 \nu ^2-4 \nu +4}+3 \nu \right)^{3/2} \left(\sqrt{18 \nu ^2+\left(6
   \sqrt{9 \nu ^2-4 \nu +4}-4\right) \nu +4}-2\right)}\,.
 \]
 Therefore, estimate \eqref{eq:theestimate} is proved with
 \[
  C(m)\;=\;\frac{\,C_3}{\,280\,}\,\mu^3\,A(r_\mathrm{max})\,,
 \]
 the dependence on $m$ being present both in $\mu$ and in $A(r_\mathrm{max})$, which are both continuously and strictly \emph{decreasing} in $m$. Thus, upon evaluating when $m=m^*$ the constants $\mu\approx 1.86$, $\nu\approx 0.13$, $r_\mathrm{max}\approx 1.57$, and $A(r_\mathrm{max})\approx 6.22$, we find $C(m^*)\approx 0.87$ and we conclude that
 \[
  \|\mathcal{K}_3^{(m)}\|\;<\;C(m)\;\leqslant\;C(m^*)\;<\;1\,,
 \]
 which completes the proof.
 \end{proof}

\section{Existence of eigenvalues}\label{sec:existence_of_EV}

In this Section we discuss the existence of eigenvalues for $H_\alpha$ and complete the proof Theorem \ref{thm:EVexist}. Only the case $\alpha<0$ is relevant, for we know already from Theorem \ref{thm:ess-spec} that $\sigma_{\mathrm{disc}}(H_\alpha)$ is empty when $\alpha\geqslant 0$.

In this regime a variational argument provides a manageable condition on the charge operator $T_1$ which is sufficient for the existence of eigenvalues below the bottom of the essential spectrum.

\begin{lemma}[Variational Lemma]\label{lem:var_lem}
 Let $m>m^*$ and $\alpha<0$. 
Assume that there exists $\widetilde{\xi}_\circ\in H^{\frac{1}{2}}(\mathbb{R}^3)$ such that
 \begin{equation}
  \varepsilon_\circ\;:=\;\frac{\;\langle \widetilde{\xi}_\circ,T_1 \widetilde{\xi}_\circ\rangle_{H^{\frac{1}{2}},H^{-\frac{1}{2}}}}{\|\widetilde{\xi}_\circ\|^2_{L^2}}\;\in\;\big[2\pi^2\sqrt{1-\Lambda(m)^2},2\pi^2\big)\,.
 \end{equation}
Then $H_\alpha$ admits a negative eigenvalue $E_\circ$ with
\begin{equation}\label{eq:whereisE0}
 E_\circ\;\in\;\Big[-\frac{\alpha^2}{4\pi^4(1-\Lambda(m)^2)},-\frac{\alpha^2}{\varepsilon_\circ^2}\Big] .
\end{equation}
\end{lemma}

\begin{proof}
 Let $\lambda_\circ:=\alpha^2/\varepsilon_\circ^2$ and $\widehat{\xi}_\circ(p):=\widehat{\widetilde{\xi}}_\circ(p/\sqrt{\lambda_\circ})$.
 Then $u_{\xi_\circ}^{\lambda_\circ}\in\mathcal{D}[H_\alpha]$ and
 \[
   H_\alpha[u_{\xi_\circ}^{\lambda_\circ}]\;=\;-\lambda_\circ\,\|u_{\xi_\circ}^{\lambda_\circ}\|^2_{L^2(\mathbb{R}^3\times\mathbb{R}^3)}+2\langle \xi_\circ,(T_{\lambda_\circ}+\alpha\mathbbm{1})\xi_\circ\rangle_{H^{\frac{1}{2}}(\mathbb{R}^3),H^{-\frac{1}{2}}(\mathbb{R}^3)}\,,
 \]
as follows from \eqref{eq:DHab_form}. Moreover, owing to Lemma \ref{lem:scaling},
\[
 0\;=\;\langle \widetilde{\xi}_\circ,T_1 \widetilde{\xi}_\circ\rangle_{H^{\frac{1}{2}},H^{-\frac{1}{2}}}-\varepsilon_\circ\|\widetilde{\xi}_\circ\|^2_{L^2}\;=\;\langle \xi_\circ,(T_{\lambda_\circ}+\alpha\mathbbm{1})\xi_\circ\rangle_{H^{\frac{1}{2}},H^{-\frac{1}{2}}}\,.
\]
Thus,
\[
 \frac{H_\alpha[u_{\xi_\circ}^{\lambda_\circ}]}{\;\;\|u_{\xi_\circ}^{\lambda_\circ}\|^2_{L^2}}\;=\;-\lambda_\circ\;\in\;\Big[-\frac{\alpha^2}{4\pi^4(1-\Lambda(m)^2)},-\frac{\alpha^2}{4\pi^4}\Big)\,,
\]
implying that there exists a normalised element $g_\circ:= u_{\xi_\circ}^{\lambda_\circ}/\|u_{\xi_\circ}^{\lambda_\circ}\|_{L^2}$ in $\mathcal{D}[H_\alpha]$ such that $H_\alpha[g_\circ]=-\lambda_\circ<\inf\sigma_{\mathrm{ess}}(H_\alpha)$. Then, by the min-max principle, there exists an eigenvalue $E_\circ$ of $H_\alpha$ below $-\lambda_\circ$ and \eqref{eq:whereisE0} follows.
\end{proof}

For the proof of Theorem \ref{thm:EVexist} we shall apply the above variational argument to the angular sector $\ell=1$ for $T_1$. Our choice of the trial function is inspired by the discussion presented in \cite{michelangeli-schmidbauer-2013} (see Fig.~7 therein).

We consider trial functions $\xi_\star$ of the form
\begin{equation}\label{eq:defxtrialgeneric}
 \widehat{\xi}_\star(p)\;:=\;f_\star(|p|)\,Y_{1,a}(\Omega_p)
\end{equation}
for some $f_\star\in L^2(\mathbb{R}^+,r^2\sqrt{r^2+1}\ud r)$ and some $a\in\{-1,0,1\}$, where we used polar coordinates $p\equiv|p|\Omega_p$ and $Y_{1,a}$ is the corresponding spherical harmonic in the $\ell=1$ sector. By means of \eqref{eq:Tform_decomp_1}-\eqref{eq:Tform_decomp_2} we compute
\begin{equation}\label{eq:T1onGenericTrial}
  \begin{split}
   \langle \xi_\star,&T_{1}\,\xi_\star\rangle_{H^{\frac{1}{2}},H^{-\frac{1}{2}}}\;=\;\Phi_0[f_\star]+\Psi_{0,1}[f_\star] \\
   &=\;2\pi^2\!\int_0^{+\infty}\!\!\ud r\,r^2\sqrt{\nu r^2+1}\,|f_\star(r)|^2 \\
   &\qquad +2\pi\!\int_0^{+\infty}\!\!\ud r\int_0^{+\infty}\!\!\ud r'\,\overline{f_\star(r)}\,f_\star(r')\int_{-1}^1\ud y\,\frac{r^2\,r'^2\,y}{r^2+r'^2+\mu r r' y + 1}\,.
  \end{split}
 \end{equation}

 We observe the following.

 \begin{lemma}\label{lem:expectation_is_monotone}
 If the function $f_\star$ chosen in \eqref{eq:defxtrialgeneric} does not depend on $m$, then the map 
 \[
  m\;\mapsto\;\mathcal{B}(m)\;:=\;\frac{\;\langle \xi_\star,T_{1}\,\xi_\star\rangle_{H^{\frac{1}{2}},H^{-\frac{1}{2}}}}{\|\xi_\star\|_{L^2}^2}
 \]
 defined by \eqref{eq:T1onGenericTrial}
 is continuous and strictly monotone increasing in $m$.
 \end{lemma}

 \begin{proof}
 Since $\xi_\star$ is independent of $m$, it suffices to consider the map $m\mapsto \langle \xi_\star,T_{1}\,\xi_\star\rangle_{H^{\frac{1}{2}},H^{-\frac{1}{2}}}$, in which case the statement follows at once from Lemma \ref{lem:T_monotonicity}, keeping into account that $\xi_\star\in H^{1/2}_{\ell=1}(\mathbb{R}^3)$.
%
%
 \end{proof}

 We then make the following choice:
 \begin{equation}\label{eq:choice}
  f_\star(r)\;:=\;\frac{e^{-b r^2}}{\,r\ln(r+a)}\,,\qquad a=1.2\,,\quad b=0.05\,.
 \end{equation}
 With this choice, the solution to the equation $\mathcal{B}(m)=2\pi^2$ is unique, owing to Lemma \ref{lem:expectation_is_monotone}, and is given by $m=m_\star$, where 
 \begin{equation}\label{eq:defmsubstar}
  m_\star\;\approx\;(8.587)^{-1}\qquad\quad (\,\mathcal{B}(m_\star)=2\pi^2\,)
 \end{equation}
 (see Figure \ref{fig:intersectionEV}).
 Based again on Lemma \ref{lem:expectation_is_monotone}, we conclude that $\mathcal{B}(m)<2\pi^2$ for $m\in(m^*,m_\star)$. Thus, in such a regime of masses, the quantity
  \[
\frac{\;\langle \xi ,T_1\xi\rangle_{H^{\frac{1}{2}},H^{-\frac{1}{2}}}}{\|\xi\|_{L^2(\mathbb{R}^3)}^2}
\]
 can be made strictly smaller than $2\pi^2$ for suitable $\xi$'s. Lemma \ref{lem:var_lem} then allows us to conclude that in the considered range of masses the operator $H_\alpha$ admits eigenvalues below the bottom of the essential spectrum.

\begin{figure}
\includegraphics[width=7cm]{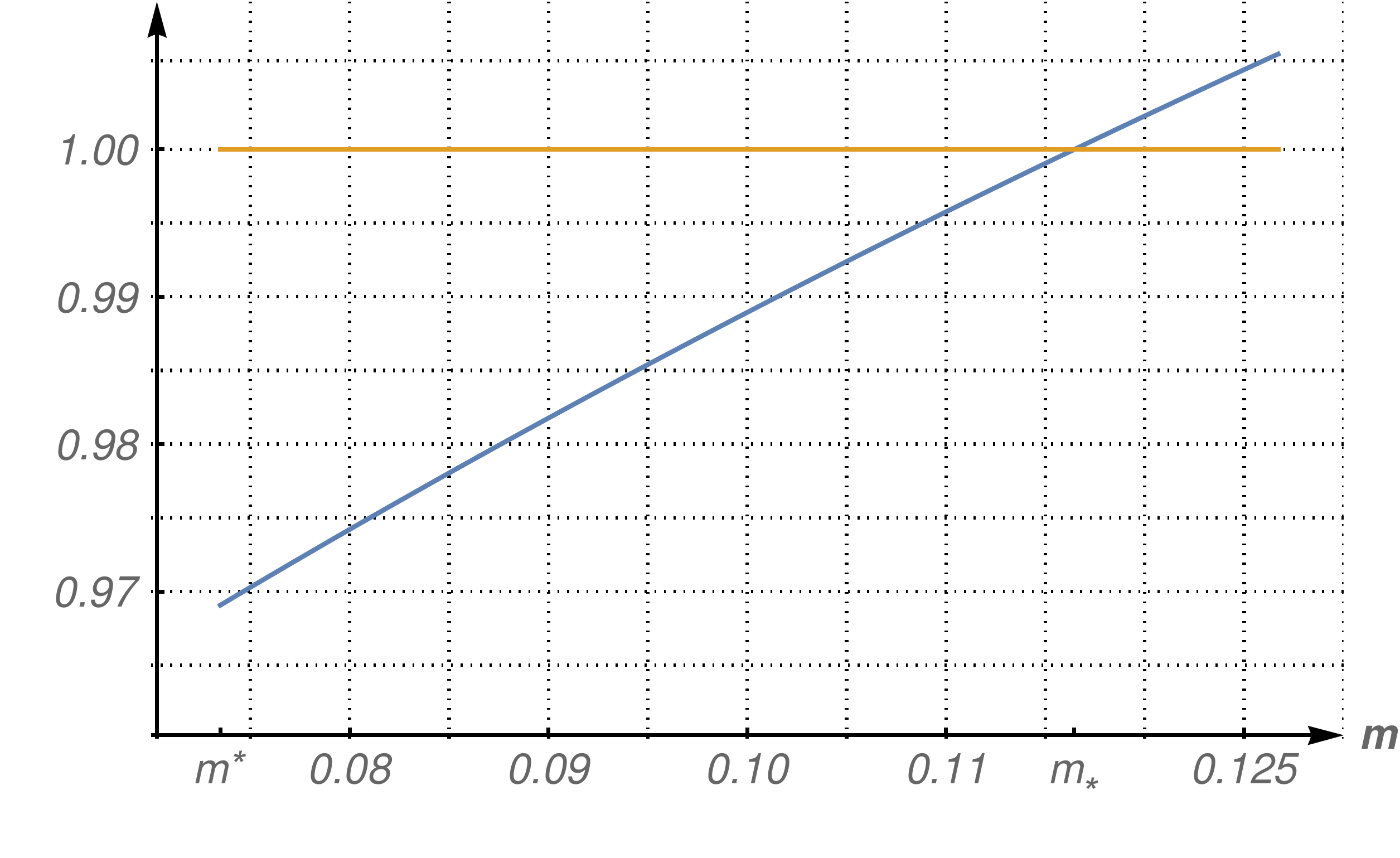}
\caption{Plot of the function $m\mapsto\frac{\mathcal{B}(m)}{2\pi^2}$ defined by \eqref{eq:T1onGenericTrial} with the choice \eqref{eq:choice} (blue curve) and intersection at $m=m_\star$ with the threshold value 1 (orange line).}\label{fig:intersectionEV}
\end{figure}

 \begin{remark}
 The mass window $(m^*,m_\star)$ that we could cover with the reasoning above is numerically the same, and in fact slightly larger, than the interval $(m^*,m^{**})$ in which bound states of $H_\alpha$ for $\alpha<0$ were already believed to exist, as indicated by the numerical evidence of \cite{michelangeli-schmidbauer-2013}. Here the threshold $m^{**}\approx (8.62)^{-1}$ is defined as follows. One can see \cite[Appendix A]{CDFMT-2015} that for $m>0$ and $s\in[0,1]$ the integral equation
 \begin{equation}\label{eq:s-integral-equation}
\pi\sqrt{{\textstyle\frac{\,m(m+2)}{(m+1)^2}}}+\int_{-1}^1\!\ud y\,y\int_0^{+\infty}\!\!\ud r\,\frac{r^s}{\,r^2+1+\frac{2}{m+1}ry}\;=\;0
\end{equation}
 defines a continuous and monotone increasing function $s\mapsto m(s)$ with $m(0)=m^*\approx (13.607)^{-1}$ given by 
 $\Lambda(m^*)=1$ in \eqref{eq:Lambdam} and $m(1)=m^{**}\approx (8.62)^{-1}$. The exact implicit formula defining $m^{**}$ is therefore
 \begin{equation}\label{eq:m**root}
\pi\sqrt{{\textstyle\frac{\,m(m+2)}{(m+1)^2}}}+\int_{-1}^1\!\ud y\,y\int_0^{+\infty}\!\!\ud r\,\frac{r}{\,r^2+1+\frac{2}{m+1}ry}\;=\;0\,.
\end{equation}
 \end{remark}

 \begin{proof}[Proof of Theorem \ref{thm:EVexist}]
 The form of the eigenfunctions is proved in Lemma \ref{lem:reduction_lemma}. The absence of eigenstates with charge of definite angular symmetry is proved in Lemma \ref{lem:only_odd_symm} for even $\ell$'s, and in Proposition \ref{prop:noEVodd3} for $\ell=3,5,7,\dots$ This establishes part (i) of the Theorem. The  existence of eigenstates with charge of angular symmetry $\ell=1$ and sufficiently small masses, including the identification of the threshold $m_\star$, is covered by the variational argument of this Section. The absence of eigenstates with charge of angular symmetry $\ell=1$ and sufficiently large masses, including the identification of the threshold $M_\star$, is proved in Proposition \ref{prop:noEV}. This establishes also part (ii) and (iii) of the Theorem.
 \end{proof}

\section{Monotonicity of eigenvalues}

In this Section we prove Theorem \ref{thm:EV_monotonicity}.

To this aim, let us slightly enrich our notation so as to include explicitly the mass dependence. With self-explanatory meaning we shall then use the symbols $H_\alpha^{(m)}$, $T_\lambda^{(m)}$, $T_{\lambda,\ell}^{(m)}$, and $u_{\xi,m}^{\lambda}$.

Moreover, at any fixed $m>m^*$ we shall enumerate in increasing order, counting the multiplicity, the (finitely many, owing to Theorem \ref{thm:EVexist}) eigenvalues of $H_\alpha^{(m)}$ below $\inf\sigma_{\mathrm{ess}}(H_\alpha^{(m)})$  as
\[
 -\lambda_1(m)\;\leqslant\;-\lambda_2(m)\;\cdots\;\leqslant\;-\lambda_{N_m}(m)\;<\;-\frac{\alpha^2}{4\pi^4}\,,
\]
where $N_m<+\infty$ is their total number.

%
%

We first establish the following useful property.

\begin{lemma}\label{lem:linear_independence}
 For given $N\in\mathbb{N}$, $\lambda_1,\dots,\lambda_N>0$, $\xi_1,\dots\xi_N\in H^{\frac{1}{2}}(\mathbb{R}^3)$, and $m>0$ (the $\lambda_j$'s and the $\xi_j$'s being not necessarily all distinct), assume that the functions $u_{\xi_1,m}^{\lambda_1},\cdots,u_{\xi_N,m}^{\lambda_N}$ are mutually orthogonal and hence
 \[
  \dim \mathrm{span}\big\{u_{\xi_1,m}^{\lambda_1},\cdots,u_{\xi_N,m}^{\lambda_N} \big\}\;=\; N\,.
 \]
 Let $m'>0$ such that  $\frac{|m-m'|}{m'(m+1)}<1$. Then
 \[
  \dim \mathrm{span}\big\{u_{\xi_1,m'}^{\lambda_1},\cdots,u_{\xi_N,m'}^{\lambda_N} \big\}\;=\; N\,.
 \]
 \end{lemma}

 \begin{proof}
 Clearly, $\dim \mathrm{span}\big\{\widehat{u_{\xi_1,m}^{\lambda_1}},\dots,\widehat{u_{\xi_N,m}^{\lambda_N}} \big\}=N$ by unitarity of the Fourier transform in $L^2_\mathrm{f}(\mathbb{R}^3\times\mathbb{R}^3)$, and the thesis is equivalent to $\dim \mathrm{span}\big\{\widehat{u_{\xi_1,m'}^{\lambda_1}},\dots,\widehat{u_{\xi_N,m'}^{\lambda_N}} \big\}=N$. Let the map
 \[
  A\;:\;\mathrm{span}\big\{\widehat{u_{\xi_1,m}^{\lambda_1}},\dots,\widehat{u_{\xi_N,m}^{\lambda_N}} \big\}\;\longrightarrow\;\mathrm{span}\big\{\widehat{u_{\xi_1,m'}^{\lambda_1}},\dots,\widehat{u_{\xi_N,m'}^{\lambda_N}}\big\}
 \]
 be defined by
 \[
  A\,\widehat{u_{\xi_j,m}^{\lambda_j}}\;:=\;\widehat{u_{\xi_j,m'}^{\lambda_j}}\qquad j\in\{1,\dots,N\}
 \]
 and extended by linearity. Writing $\mu'=\frac{2}{m'+1}$ (and $\mu=\frac{2}{m+1}$) one has
 \[
  \widehat{u_{\xi_j,m'}^{\lambda_j}}(p,q)\;=\;1+(\mu-\mu')\frac{p\cdot q}{p^2+q^2+\mu' p\cdot q+\lambda_j}\,\widehat{u_{\xi_j,m}^{\lambda_j}}(p,q)\,,
 \]
 as follows directly from \eqref{eq:u_xi}. Thus,
 \[
  (A-\mathbbm{1})\widehat{u_{\xi_j,m}^{\lambda_j}}\;=\;(\mu-\mu')\frac{p\cdot q}{p^2+q^2+\mu' p\cdot q+\lambda_j}\,\widehat{u_{\xi_j,m}^{\lambda_j}}
 \]
 and obviously $\big|\frac{(\mu-\mu')\,p\cdot q}{p^2+q^2+\mu' p\cdot q+\lambda_j}\big|\leqslant \frac{|\mu-\mu'|}{2-\mu'}$, whence also
 \[
  \big\|(A-\mathbbm{1})\widehat{u_{\xi_j,m}^{\lambda_j}}\big\|_{\cH}\;\leqslant\;\frac{|\mu-\mu'|}{2-\mu'}\,\big\|\widehat{u_{\xi_j,m}^{\lambda_j}}\big\|_{\cH} 
 \]
 uniformly in $p,q,j$. It is straightforward to see that the latter inequality implies
 \[
  \|A-\mathbbm{1}\|\;\leqslant\;\frac{|\mu-\mu'|}{2-\mu'}\;=\;\frac{|m-m'|}{m'(m+1)}\,.
 \]
 Indeed, with the compact notation
 \[
  e_j\;:=\;\widehat{u_{\xi_j,m}^{\lambda_j}}\,\big/ \|\widehat{u_{\xi_j,m}^{\lambda_j}}\big\|_{\cH} \qquad j\in\{1,\dots,N\}\,,
 \]
 one sees that for a generic element $\sum_{j=1}^N c_j e_j$ ($c_j\in\mathbb{C}$) in the span of $\{e_1,\dots,e_n\}$ one has
 \[
  \begin{split}
  \|(A-\mathbbm{1})&{\textstyle \sum_{j=1}^N c_j e_j}\|_{\cH}^2\;=\;\sum_{j,k=1}^N \overline{c_j}\,c_k\langle (A-\mathbbm{1})e_j,(A-\mathbbm{1})e_k\rangle_{\cH} \\
  &\leqslant\;\textstyle \big(\frac{|m-m'|}{m'(m+1)}\big)^2\sum_{j,k=1}^N|c_j|\,|c_k| \;\leqslant\;\textstyle \big(\frac{|m-m'|}{m'(m+1)}\big)^2\sum_{j=1}^N|c_j|^2\,,
  \end{split}
 \]
 whence the conclusion. Last, since by assumption $\frac{|m-m'|}{m'(m+1)}<1$, then the map $A=\mathbbm{1}+(A-\mathbbm{1})$ is invertible, implying at once that $\dim \mathrm{span}\big\{\widehat{u_{\xi_1,m'}^{\lambda_1}},\dots,\widehat{u_{\xi_N,m'}^{\lambda_N}}\big\}=N$, which concludes the proof.
 \end{proof}

\begin{proof}[Proof of Theorem \ref{thm:EV_monotonicity}]
Let us proceed inductively, fixing two masses $m_1,m_2>m^*$, with $m_1<m_2$, and assuming non-restrictively that such masses are sufficiently close, in the quantitative sense $\frac{|m_2-m_1|}{m_1(m_2+1)}<1$.

First, we prove that the \emph{lowest} eigenvalue of $H_\alpha^{(m)}$ is strictly monotone increasing with $m$ in the sense that if both Hamiltonians $H_\alpha^{(m_1)}$ and $H_\alpha^{(m_2)}$ have non-empty discrete spectrum, then $-\lambda_1(m_1)<-\lambda_1(m_2)$.

Owing to Lemmas \ref{lem:reduction_lemma} and \ref{lem:only_odd_symm},
an eigenfunction corresponding to the eigenvalue $-\lambda_1(m_2)$ has necessarily the form (of a multiple of) $u_{\xi,m_2}^{\lambda_1(m_2)}$ for some $\xi\in H_-^{1/2}(\mathbb{R}^3)$, and 
\[
 \begin{split}
  (H^{(m_2)}_\alpha+\lambda_1(m_2)\mathbbm{1})\big[u_{\xi,m_2}^{\lambda_1(m_2)}\big]\;=\;\langle\xi,(T_{\lambda_1(m_2)}^{(m_2)}+\alpha\mathbbm{1})\xi\rangle_{H^{\frac{1}{2}},H^{-\frac{1}{2}}}\;=\;0\,.
 \end{split}
\]
On the other hand, forming the function $u_{\xi,m_1}^{\lambda_1(m_2)}\in\mathcal{D}[H_\alpha^{(m_1)}]$, one has
\[
 \begin{split}
  (H^{(m_1)}_\alpha+\lambda_1(m_2)\mathbbm{1})\big[u_{\xi,m_1}^{\lambda_1(m_2)}\big]\;&=\;\langle\xi,(T_{\lambda_1(m_2)}^{(m_1)}+\alpha\mathbbm{1})\xi\rangle_{H^{\frac{1}{2}},H^{-\frac{1}{2}}} \\
  &<\;\langle\xi,(T_{\lambda_1(m_2)}^{(m_2)}+\alpha\mathbbm{1})\xi\rangle_{H^{\frac{1}{2}},H^{-\frac{1}{2}}}\;=\;0\,,
 \end{split}
\]
where the first step follows from formula \eqref{eq:DHab_form} for the evaluation of the quadratic form of $H_\alpha^{(m_1)}$, and the inequality in the second step follows from the monotonicity Lemma \ref{lem:T_monotonicity}. Therefore,
\[
 \frac{\;H^{(m_1)}_\alpha\big[u_{\xi,m_1}^{\lambda_1(m_2)}\big]\;}{\big\|u_{\xi,m_1}^{\lambda_1(m_2)}\big\|_{\cH}^2}\;<\;-\lambda_1(m_2)\,,
\]
whence the existence of an eigenvalue of $H_\alpha^{(m_1)}$ strictly below $-\lambda_1(m_2)$, and then the conclusion $-\lambda_1(m_1)<-\lambda_1(m_2)$.

Next, by induction, let us assume that for the first $k$ eigenvalues (counting the multiplicity) of $H_\alpha^{(m_2)}$ one has $-\lambda_j(m_1)<-\lambda_j(m_2)$ for any $j\in\{1,\dots,k\}$ and let us further assume that $k< N_{m_2}$, that is, $H_\alpha^{(m_2)}$ admits one further $(k+1)$-th eigenvalue $-\lambda_{k+1}(m_2)$. Let us denote the corresponding eigenfunctions with
\[
 u_{\xi_1,m_2}^{\lambda_1(m_2)}\,,
 \,\dots\,,\,u_{\xi_k,m_2}^{\lambda_k(m_2)}\,,\, u_{\xi_{k+1},m_2}^{\lambda_{k+1}(m_2)}
\]
for some charges $\xi_1,\dots,\xi_k,\xi_{k+1}\in H_-^{1/2}(\mathbb{R}^3)$, assuming that such eigenfunctions are chosen so as to be mutually orthogonal, which is always possible also  in case of degeneracy.
Since in addition $\frac{|m_2-m_1|}{m_1(m_2+1)}<1$, then we can apply Lemma \ref{lem:linear_independence} and deduce that the functions
\[
 u_{\xi_1,m_1}^{\lambda_1(m_2)}\,,
 \,\dots\,,\,u_{\xi_{k+1},m_1}^{\lambda_{k+1}(m_2)}
\]
are linearly independent too. On each such function one has
\[
 \begin{split}
  (H^{(m_1)}_\alpha+\lambda_{j}&(m_2)\mathbbm{1})\big[u_{\xi_{j},m_1}^{\lambda_{j}(m_2)}\big]\;=\;\langle\xi_{j},(T_{\lambda_{j}(m_2)}^{(m_1)}+\alpha\mathbbm{1})\xi_{j}\rangle_{H^{\frac{1}{2}},H^{-\frac{1}{2}}} \\
  &<\;\langle\xi_{j},(T_{\lambda_{j}(m_2)}^{(m_2)}+\alpha\mathbbm{1})\xi_{j}\rangle_{H^{\frac{1}{2}},H^{-\frac{1}{2}}} \\
  &=\;(H^{(m_2)}_\alpha+\lambda_{j}(m_2)\mathbbm{1})\big[u_{\xi_{j},m_2}^{\lambda_{j}(m_2)}\big]\;=\;0 \qquad j\in\{1,\dots,k+1\}
 \end{split}
\]
%
having used again \eqref{eq:DHab_form} in the first step and third step, and Lemma \ref{lem:T_monotonicity} in the second step. This shows that there exists a $(k+1)$-dimensional subspace of $\mathcal{D}[H_\alpha^{(m_1)}]$, the space spanned by $u_{\xi_1,m_1}^{\lambda_1(m_2)},\dots,u_{\xi_{k+1},m_1}^{\lambda_{k+1}(m_2)}$,
on the elements of which the normalised expectations of $H_\alpha^{(m_1)}$ are strictly below $-\lambda_{k+1}(m_2)$, i.e., the largest among all the $-\lambda_{j}(m_2)$'s. This implies that $H_\alpha^{(m_1)}$ admits at least $k+1$ eigenvalues strictly below $-\lambda_{k+1}(m_2)$, and in particular $-\lambda_{k+1}(m_1)<-\lambda_{k+1}(m_2)$.

By induction, the proof is then completed.
\end{proof}

%



\begin{thebibliography}{10}

\bibitem{Bethe_Peierls-1935}
{\sc H.~Bethe and R.~Peierls}, {\em {Quantum Theory of the Diplon}},
  Proceedings of the Royal Society of London. Series A, Mathematical and
  Physical Sciences, 148 (1935), pp.~146--156.

\bibitem{Bethe_Peierls-1935-np}
{\sc H.~A. Bethe and R.~Peierls}, {\em {The Scattering of Neutrons by
  Protons}}, Proceedings of the Royal Society of London. Series A, Mathematical
  and Physical Sciences, 149 (1935), pp.~176--183.

\bibitem{Braaten-Hammer-2006}
{\sc E.~Braaten and H.-W. Hammer}, {\em {Universality in few-body systems with
  large scattering length}}, Physics Reports, 428 (2006), pp.~259--390.

\bibitem{Castin-Werner-2011_-_review}
{\sc Y.~Castin and F.~Werner}, {\em {The Unitary Gas and its Symmetry
  Properties}}, in {The BCS-BEC Crossover and the Unitary Fermi Gas},
  W.~Zwerger, ed., vol.~836 of {Lecture Notes in Physics}, Springer Berlin
  Heidelberg, 2012, pp.~127--191.

\bibitem{CDFMT-2012}
{\sc M.~Correggi, G.~Dell'Antonio, D.~Finco, A.~Michelangeli, and A.~Teta},
  {\em {Stability for a system of {$N$} fermions plus a different particle with
  zero-range interactions}}, Rev. Math. Phys., 24 (2012), pp.~1250017, 32.

\bibitem{CDFMT-2015}
\leavevmode\vrule height 2pt depth -1.6pt width 23pt, {\em {A Class of
  Hamiltonians for a Three-Particle Fermionic System at Unitarity}},
  Mathematical Physics, Analysis and Geometry, 18 (2015).

\bibitem{DFT-1994}
{\sc G.~F. Dell'Antonio, R.~Figari, and A.~Teta}, {\em {Hamiltonians for
  systems of {$N$} particles interacting through point interactions}}, Ann.
  Inst. H. Poincar{\'e} Phys. Th{\'e}or., 60 (1994), pp.~253--290.

\bibitem{Endo-Naidon-Ueda-2011}
{\sc S.~Endo, P.~Naidon, and M.~Ueda}, {\em {Universal Physics of 2+1 Particles
  with Non-Zero Angular Momentum}}, Few-Body Systems, 51 (2011), pp.~207--217.

\bibitem{Finco-Teta-2012}
{\sc D.~Finco and A.~Teta}, {\em {Quadratic forms for the fermionic unitary gas
  model}}, Rep. Math. Phys., 69 (2012), pp.~131--159.

\bibitem{GMO-KVB2017}
{\sc M.~Gallone, A.~Michelangeli, and A.~Ottolini}, {\em
  {Kre{\u\i}n-Vi\v{s}ik-Birman self-adjoint extension theory revisited}}, SISSA
  preprint 25/2017/MATE (2017).

\bibitem{Gradshteyn-tables-of-integrals-etc}
{\sc I.~S. Gradshteyn and I.~M. Ryzhik}, {\em {Table of integrals, series, and
  products}}, Elsevier/Academic Press, Amsterdam, eighth~ed., 2015.
\newblock Translated from the Russian, Translation edited and with a preface by
  Daniel Zwillinger and Victor Moll, Revised from the seventh edition
  [MR2360010].

\bibitem{Kartavtsev-Malykh-2016-proc}
{\sc O.~I. Kartavtsev and A.~V. Malykh}, {\em {Universal descritpion of three
  two-species particles}}, in {Proc. of the 4th South Africa-JINR Symposium
  ``Few to Many Body Systems: Models, Methods, and Applications''}, A.~V.
  Malykh, ed., pp.~23--29.

\bibitem{Kartavtsev-Malykh-2007}
\leavevmode\vrule height 2pt depth -1.6pt width 23pt, {\em {Low-energy
  three-body dynamics in binary quantum gases}}, Journal of Physics B: Atomic,
  Molecular and Optical Physics, 40 (2007), p.~1429.

\bibitem{Kartavtsev-Malykh-2016}
\leavevmode\vrule height 2pt depth -1.6pt width 23pt, {\em {Universal
  description of three two-component fermions}}, EPL, 115 (2016), p.~36005.

\bibitem{MO-2017}
{\sc A.~Michelangeli and A.~Ottolini}, {\em {Multiplicity of self-adjoint
  realisations of the (2+1)-fermionic model of Ter-Martirosyan--Skornyakov
  type}}, to appear in Rep. Math. Phys. -- SISSA preprint 65/2016/MATE,
  (2017).

\bibitem{MO-2016}
\leavevmode\vrule height 2pt depth -1.6pt width 23pt, {\em {On point
  interactions realised as {T}er-{M}artirosyan-{S}kornyakov {H}amiltonians}},
  Rep. Math. Phys., 79 (2017), pp.~215--260.

\bibitem{MP-2015-2p2}
{\sc A.~Michelangeli and P.~Pfeiffer}, {\em {Stability of the (2+2)-fermionic
  system with zero-range interaction}}, Journal of Physics A: Mathematical and
  Theoretical, 49 (2016), p.~105301.

\bibitem{michelangeli-schmidbauer-2013}
{\sc A.~Michelangeli and C.~Schmidbauer}, {\em {Binding properties of the
  (2+1)-fermion system with zero-range interspecies interaction}}, Phys. Rev.
  A, 87 (2013), p.~053601.

\bibitem{Minlos-1987}
{\sc R.~A. Minlos}, {\em {On the point interaction of three particles}}, in
  {Applications of selfadjoint extensions in quantum physics ({D}ubna, 1987)},
  vol.~324 of {Lecture Notes in Phys.}, Springer, Berlin, 1989, pp.~138--145.

\bibitem{Minlos-TS-1994}
\leavevmode\vrule height 2pt depth -1.6pt width 23pt, {\em {On pointlike
  interaction between {$N$} fermions and another particle}}, in {Proceedings of
  the Workshop on Singular Schr{\"o}dinger Operators, Trieste 29 September - 1
  October 1994}, A.~Dell'Antonio, R.~Figari, and A.~Teta, eds., {ILAS/FM-16},
  1995.

\bibitem{Minlos-2011-preprint_May_2010}
\leavevmode\vrule height 2pt depth -1.6pt width 23pt, {\em {On point-like
  interaction between {$n$} fermions and another particle}}, Mosc. Math. J., 11
  (2011), pp.~113--127, 182.

\bibitem{Minlos-2012-preprint_30sett2011}
\leavevmode\vrule height 2pt depth -1.6pt width 23pt, {\em {On Point-like
  Interaction between Three Particles: Two Fermions and Another Particle}},
  ISRN Mathematical Physics, 2012 (2012), p.~230245.

\bibitem{Minlos-RusMathSurv-2014}
\leavevmode\vrule height 2pt depth -1.6pt width 23pt, {\em {A system of three
  pointwise interacting quantum particles}}, Uspekhi Mat. Nauk, 69 (2014),
  pp.~145--172.

\bibitem{Minlos-2012-preprint_1nov2012}
\leavevmode\vrule height 2pt depth -1.6pt width 23pt, {\em {On point-like
  interaction of three particles: two fermions and another particle. {II}}},
  Mosc. Math. J., 14 (2014), pp.~617--637, 642--643.

\bibitem{Minlos-Faddeev-1961-2}
{\sc R.~A. Minlos and L.~D. Faddeev}, {\em {Comment on the problem of three
  particles with point interactions}}, Soviet Physics JETP, 14 (1962),
  pp.~1315--1316.

\bibitem{Minlos-Faddeev-1961-1}
\leavevmode\vrule height 2pt depth -1.6pt width 23pt, {\em {On the point
  interaction for a three-particle system in quantum mechanics}}, Soviet
  Physics Dokl., 6 (1962), pp.~1072--1074.

\bibitem{Minlos-Shermatov-1989}
{\sc R.~A. Minlos and M.~K. Shermatov}, {\em {Point interaction of three
  particles}}, Vestnik Moskov. Univ. Ser. I Mat. Mekh.,  (1989), pp.~7--14, 97.

\bibitem{pethick02}
{\sc C.~J. Pethick and H.~Smith}, {\em {Bose--Einstein Condensation in Dilute
  Gases}}, Cambridge University Press, second~ed., 2008.
\newblock Cambridge Books Online.

\bibitem{Petrov-Ultracoldgases-LH2010}
{\sc D.~S. Petrov}, {\em {The few-atom problem}}, in {Many-Body Physics with
  Ultracold Gases (Les Houches 2010), Lecture Notes of the Les Houches Summer
  School vol.~94}, Oxford Univ. Press, Oxford, 2013, pp.~109--160.

\bibitem{schmu_unbdd_sa}
{\sc K.~Schm{\"u}dgen}, {\em {Unbounded self-adjoint operators on {H}ilbert
  space}}, vol.~265 of {Graduate Texts in Mathematics}, Springer, Dordrecht,
  2012.

\bibitem{TMS-1956}
{\sc G.~V. Skornyakov and K.~A. Ter-Martirosyan}, {\em {Three Body Problem for
  Short Range Forces. I. Scattering of Low Energy Neutrons by Deuterons}}, Sov.
  Phys. JETP, 4 (1956), pp.~648--661.

\bibitem{Teta-1989}
{\sc A.~Teta}, {\em {Quadratic forms for singular perturbations of the
  {L}aplacian}}, Publ. Res. Inst. Math. Sci., 26 (1990), pp.~803--817.

\bibitem{Yoshitomi_MathSlov2017}
{\sc K.~Yoshitomi}, {\em {Finiteness of the discrete spectrum in a three-body
  system with point interaction}}, Math. Slovaca, 67 (2017), pp.~1031--1042.

\end{thebibliography}

\def\cprime{$'$}

\end{document}